\documentclass[12pt]{article}

\usepackage[authoryear]{natbib}
\usepackage{amsmath,amssymb,amsthm}
\numberwithin{equation}{section}
\usepackage{graphicx}
\usepackage{float}
\usepackage{multirow}
\usepackage{arydshln}
\usepackage{booktabs}
\usepackage{etoolbox}

\usepackage{paralist}

\usepackage[T1]{fontenc}
\usepackage[colorlinks,citecolor=blue,urlcolor=blue]{hyperref}

\newtheorem{theorem}{Theorem}
\newtheorem{lemma}{Lemma}
\newtheorem{corollary}{Corollary}

\newcommand{\bsb}{\boldsymbol}

\DeclareMathOperator{\vect}{\mbox{vec}\,}

\newcommand{\bsbX}{{\boldsymbol{X}}}
\newcommand{\bsbx}{{\boldsymbol{x}}}
\newcommand{\bsby}{{\boldsymbol{y}}}
\newcommand{\bsbY}{{\boldsymbol{Y}}}
\newcommand{\bsbb}{{\boldsymbol{\beta}}}

\newcommand{\bsbT}{{\boldsymbol{T}}}

\newcommand{\bsbI}{{\boldsymbol{I}}}

\newcommand{\bsbR}{{\boldsymbol{R}}}
\newcommand{\bsbZ}{{\boldsymbol{Z}}}
\newcommand{\bsbO}{{\boldsymbol{O}}}
\newcommand{\bsbSig}{{\boldsymbol{\Sigma}}}
\newcommand{\bsbOmega}{{\boldsymbol{\Omega}}}
\newcommand{\bsbxi}{{\boldsymbol{\xi}}}
\newcommand{\bsbD}{{\boldsymbol{D}}}

\newcommand{\bsbU}{{\boldsymbol{U}}}
\newcommand{\bsbV}{{\boldsymbol{V}}}

\newcommand{\bsba}{{\boldsymbol{\alpha}}}
\newcommand{\bsbA}{{\boldsymbol{A}}}

\newcommand{\bsbC}{{\boldsymbol{C}}}

\newcommand{\bsbE}{{\boldsymbol{E}}}

\newcommand{\bsbS}{{\boldsymbol{S}}}

\newcommand{\bsbB}{{\boldsymbol{B}}}

\newcommand{\bsbDelta}{{\boldsymbol{\Delta}}}

\newcommand{\Proj}{{\boldsymbol{\mathcal P}}}
\newcommand{\EP}{\,\mathbb{P}}
\newcommand{\EE}{\,\mathbb{E}}

\begin{document}
\title{On Cross-validation for  Sparse  Reduced Rank Regression}
\author{Yiyuan She and Hoang Tran \\Department of Statistics, Florida State University}
\date{}
\maketitle

\begin{abstract}
In high-dimensional data analysis, regularization methods pursuing   sparsity and/or low rank have received a lot of attention recently.   To provide      a proper amount of shrinkage, it is typical to use     a grid search and a model comparison criterion to find the optimal    regularization parameters. However, we show that fixing the parameters   across all folds may result in an inconsistency issue, and  it is more appropriate to cross-validate   projection-selection patterns to obtain the best coefficient   estimate. Our     in-sample error studies in jointly sparse and rank-deficient models lead to  a new class of information criteria with four scale-free forms to bypass the estimation of the  noise level. By use of an identity, we propose a novel scale-free  calibration to help cross-validation achieve the minimax optimal error rate   non-asymptotically. Experiments  support the efficacy of the proposed methods.
\end{abstract}

\section{Background}
Modern statistical techniques heavily rely on shrinkage estimation and thus   parameter tuning  becomes an important task. This work assumes  a  multivariate   regression setup,  $\bsbY=\bsbX\bsbB^* +\bsbE$, which is  applicable to a variety of real-world applications  in machine learning, economics and genetics studies \citep{stock2002forecasting,ESL2,vounou2012sparse}. Here,  $\bsbY \in \mathbb R^{n\times m}$ denotes the response matrix, $\bsbX\in \mathbb R^{n\times p}$ is the design matrix containing $p$ predictors or features, and the noise  $\bsbE$ is assumed to be  Gaussian (or sub-Gaussian).
The $j$th row of $\bsbB^*$ contains the coefficients associated with the $j$th predictor, and the $k$th column of $\bsbB^*$ corresponds to the coefficients associated with the $k$th response in $\bsbY$. Therefore, when  a number of  features are   irrelevant to $\bsbY$,   $\bsbB^*$ has row sparsity.  In the case of    a single response, the problem becomes standard   variable selection.
But variable selection alone may not be sufficient when $m>1$. Suppose that   the first $J^*$ rows in $\bsbB^{*}$ are nonzero, i.e.,
$\bsbB^* = [\bsbb_1^* \ \cdots \ \bsbb_{J^*}^* \ \bsb0 \ \cdots \  \bsb0]^T$.
 The number of unknowns,  $m \times J^*$,  can still be large and exceed the total number of observations. In   matrix estimation, adding a low-rank constraint is   popular and effective \citep{CandPlan,Rohde11}. Specifically, suppose further that $\bsbB^{*}$ has  rank $r^*<J^*$, which means that $\bsbB^{*} = \bsbC_1^{*} \left( \bsbC_2^{*} \right)^T$ for some $\bsbC_1^* \in \mathbb{R}^{p \times r^*}$ and $\bsbC_2^* \in \mathbb{R}^{m \times r^*}$. Then, the model can be rewritten as $\bsbY = \left( \bsbX \bsbC_1^* \right) \left( \bsbC_2^* \right)^T + \bsbE$, where $\bsbX \bsbC_1^*$ is composed of $r^*$ factors which are linear combinations of   $J^*$ relevant predictors to explain \emph{all} response variables. This is an extension of sparse principal component analysis where $\bsbX$   is an identity matrix   \citep{spapca,shenhuang,pmd,johnstoneLu2009,Ma2013}. For more details on jointly rank-deficient and sparse regression, we refer to  \cite{BSW12}, \cite{ChenHuang12}, \cite{chen2012jrssb}, and  \cite{she2017selective}.

Although   joint       regularization guarantees effective dimension reduction,    it can  be challenging   to  control the amount of shrinkage adaptively in real data.
Let $\lambda$ denote the regularization parameter(s) to be tuned. In the literature, a  grid search is often used, which calls a sparse learning algorithm for every   value of $\lambda$  in a pre-specified grid. So the problem    boils down to the design of an appropriate model comparison criterion.

There are two broad classes of comparison criteria: information criteria, appending various  penalties on the model complexity to the   training error, and cross-validation (CV), which is based on data resampling. These criteria offer  {data-dependent}  (or adaptive) parameters, in contrast to  some theoretical choices  that are   derived assuming  incoherent designs.

For ordinary variable selection ($m=1$), some examples of information criteria are AIC~\citep{akaike1974new}, BIC~\citep{schwarz1978estimating}, RIC~\citep{FosterGeorge} and EBIC~\citep{chenchen}. We refer to~\cite{shao1997asymptotic} and~\cite{yang2005can} for  asymptotic studies under the classical regime with $n \rightarrow \infty$ and $p$ fixed. In practical data analysis, however,  there appears to be no clear conclusion of which criterion to use,  and some results   even seem contradictory. There is much less theoretical work on   rank selection with $p\gg n$  \citep{And99,BuneaSheWeg}, let alone for joint variable selection and rank reduction.
 Instead of trying to    figure out the appropriate  information criterion penalty   in a specific setting, many practitioners prefer   general-purpose $K$-fold cross-validation  \citep{geisser1975predictive}. This requires us to  split the dataset randomly  into $K$ subsets of roughly equal size, then to
fit,  say, a    lasso model~\citep{Tib} at any given value of $\lambda$ on the data without the $k$-th subset  ($1\leq k \leq K$), and   to evaluate
prediction error on the left-out subset. We refer to~\cite{arlot2010survey} for a modern survey of cross-validation.
It is a common  belief that the summarized  cross-validation error provides a good index of    a model's  goodness of fit. The choice of $K$ is often discretionary but can have a significant effect.

The main contributions of this paper are three-fold. First, we argue that the $K$ trainings   in the     conventional cross-validation may   be associated with   \textit{inconsistent} models,  especially when a   nonconvex penalty is in use, and  we introduce a structural  cross-validation (SCV) based on
selection-projection patterns.  Second, we develop a new class of  predictive information criteria (PIC) that can achieve  the minimax optimal error rate  without {any} design incoherence or signal strength conditions. Third,   we show that $K$-fold cross-validation is \textit{not} rate-optimal in    pursuing sparsity, and   we propose  a novel scale-free calibration of the CV error to match the optimal rate.   Experiments  show the superb performance of the proposed methods.

The majority of this work  will  assume that the model is   jointly row-sparse and rank-deficient. The benefits will be seen, for example,  in  revealing
the \textbf{additive} relationship between  degrees of freedom and inflation, which may be  difficult to perceive in single-response regression.   But all of our results and discussions  apply to      pure variable selection and low-rank modeling.

The following notation and symbols will be used. Given $\bsbB\in \mathbb R^{p\times m}$, we shall use $r(\bsbB)$ to denote the rank of $\bsbB$, $\mathcal J(\bsbB)  =\{1\le j \le p: \| \bsbB[j,]\|_2\ne 0\}$ to denote the row support of $\bsbB$, and $ J(\bsbB) = | \mathcal  J(\bsbB)|$. We use  $\bsbB[\mathcal I, \mathcal J]$ to denote a submatrix of $\bsbB$ with rows and columns indexed by $\mathcal I, \mathcal J$ respectively, and abbreviate $\bsbB[ 1\! : \! p, \mathcal J]$ to $\bsbB[ , \mathcal J]$. Let $\mathbb O^{m\times n}=\{\bsbA\in \mathbb R^{m\times n}: \bsbA^T \bsbA=\bsbI\}$ represent the class of orthogonal matrices. The symbol $\gtrsim$ means that the inequality holds up to multiplicative constants (similarly for $\lesssim$). We use $\| \bsbA\|_F$ to denote the Frobenius norm of $\bsbA$. Finally, $\Proj_{\bsbA}$ is the orthogonal projection matrix onto  the column space of  $ \bsbA$, i.e., $\Proj_{\bsbA}=\bsbA(\bsbA^T \bsbA)^{-} \bsbA^T$ where $^-$ is the Moore-Penrose inverse; when there is no ambiguity, we also use     $\Proj_{\bsbA}$   to denote the column space of $\bsbA$.

\section{Structural Cross-Validation }\label{sec:structural}

When using cross-validation for   parameter tuning,   one often holds $\lambda$ \textit{fixed} at a given value across all folds. 
Some recent theoretical studies partially support this idea. For instance, consider  the lasso   which minimizes $ \| \bsby - \bsbX \bsbb\|_2^2 + \lambda \| \bsbb\|_1$ with $\bsbX \in \mathbb{R}^{n \times p}$, $\bsby \in \mathbb R^n$.
Under some incoherent   conditions on the design matrix,   $\lambda$   has a {universal}  rate choice of $ O(\sigma    \sqrt{ n \log p })$; see, for example,  ~\cite{bickel09}. This rate   depends only on problem dimensions and     makes the lasso    achieve  nearly   optimal prediction error rate
 with high probability.
Similar conclusions exist for  singular-value thresholds in low-rank matrix estimation; see, for example, \cite{CandPlan} and \cite{BuneaSheWeg}.  These universal-rate results may lend some support to fixing the regularization parameter(s) in cross-validation, as long as the designs in different trainings are of     the same size.

\begin{figure}[h!]
\centering
\includegraphics[width=0.8\textwidth]{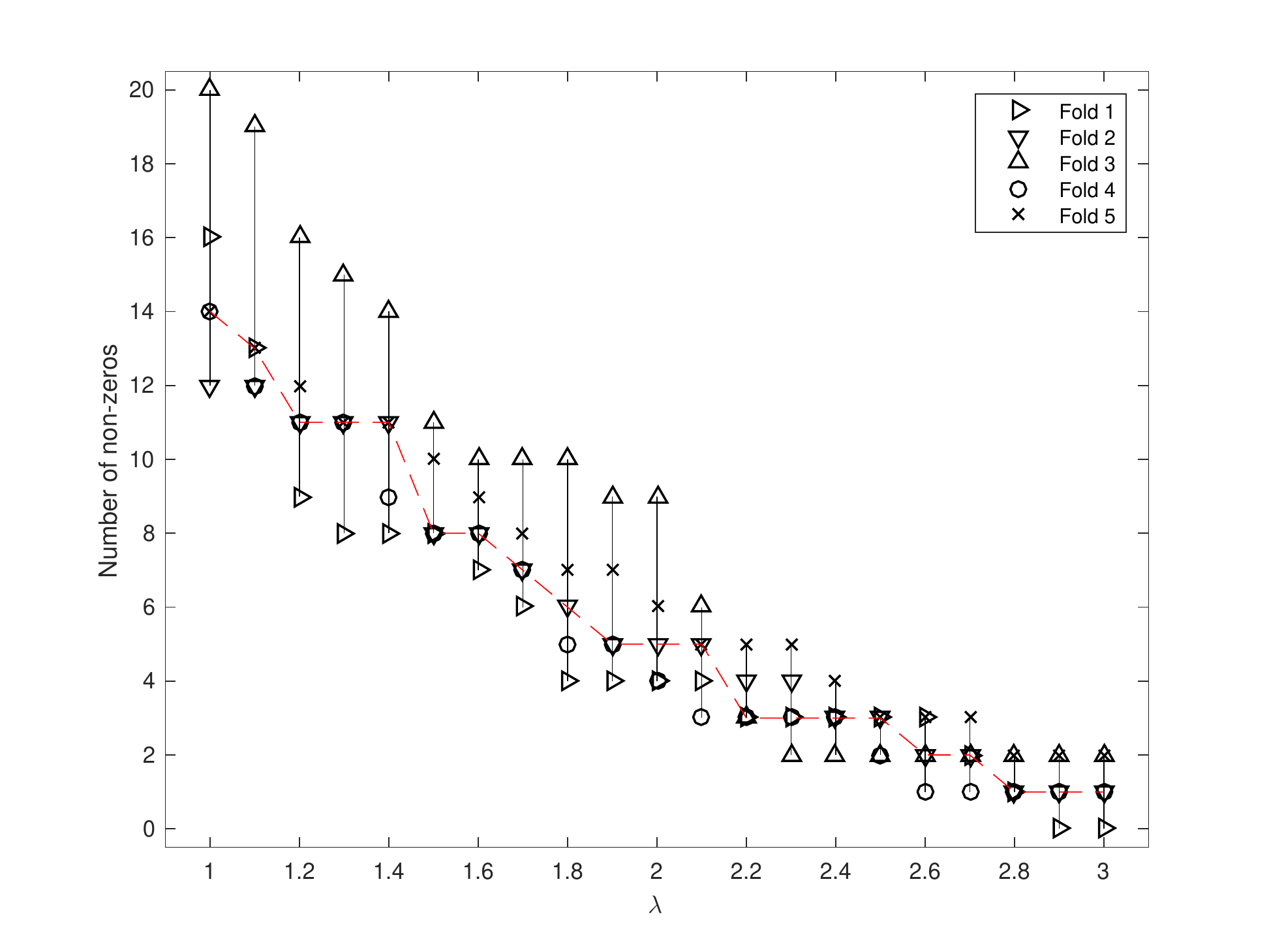}
\caption{\small Training inconsistency  in  cross-validating $\lambda$. The number of nonzero coefficient estimates  is plotted against $\lambda$ in each training (computed by lasso). Each vertical bar shows the range of cardinalities at a given $\lambda$ and the dashed line connects the medians. In an ideal situation where the $K$ fittings are regarding the same candidate model,
the vertical variation in the plot would disappear (and one would see a  solution path similar to the dashed line).  \label{fig:CV} }
\end{figure}

 However,  the rate   choice, as well as  the error bounds obtained with high probability,       is  too crude   and restrictive in real-life data analysis. Moreover, the design matrix is required to satisfy, say, the restricted eigenvalue condition~\citep{bickel09}, the restricted isometry property~\citep{candes2008restricted},   the sparse Riesz condition~\citep{ZhangHuang} or the comparison condition \citep{She2016},     all   placing stringent requirements on the Gram matrix $\bsbX^T \bsbX$ and  easy to violate in practice. There is yet no theory that is sufficiently
fine   to take all   characteristics of data     into  consideration to give a precise   formula of  $\lambda$.
The   issue is perhaps easier to understand from a Lagrangian perspective. Indeed, the     $\lambda$ in lasso is a {Lagrangian multiplier} of the   constrained   problem
$
\min_\bsbb \| \bsby - \bsbX \bsbb\|_2^2 \mbox{ s.t. } \| \bsbb \|_1 \leq c$.
In contrast to the constraint parameter   $c$ that is directly defined on $\bsbb$, $\lambda=\lambda(\bsbX, \bsby, c)$, as a dual parameter of $c$,  depends on $ \bsbX$ and  $\bsby$ as well, and   this data dependence is not just restricted to  problem's  dimensions. Therefore, on different datasets (associated with the same $\bsbb$),  the penalty parameter 
 may have to be changed to maintain the same constraint value. 
Therefore,
   with   $\lambda$ held fixed in CV, the $K$ fitted models  may not be consistent 
and the judgement based on the  total  CV-error could be       spurious.
  See Figure \ref{fig:CV} for an illustration of the issue,   where for a given value of $\lambda$  the cardinality of the trained models varies across the folds (for example,  at $\lambda=1.2$ the number of non-zeros  varies from nine to sixteen).
The issue was discussed in \cite{SheGLMTISP} and \cite{SheSpec}.

What if we were to   cross-validate $c$? This constraint form is exactly what \cite{Tib} used to  define the lasso criterion and to call cross-validation on.
But the idea   puts some limitations on computation, especially when {a nonconvex} regularization is in use. Solving a constrained problem is often not easy; for example, changing a SCAD penalty to a constrained form results in a much harder problem~\citep{fan2001variable}.   Moreover, it is well known that  the solution path associated with a nonconvex penalty or constraint is  discontinuous, and       algorithms may be easily trapped  into a local minimum. These matters only exacerbate the training inconsistency issue.

There is  a simple idea to
address the problem. Instead of finding the optimal regularization parameter, our goal should be to obtain the best coefficient matrix estimate.
So  we can cross-validate factor loadings  instead of  regularization parameters to maintain the same model in $K$ trainings and validations. Below we describe how to extract  selection-projection patterns in the jointly row-sparse and rank-deficient setup  given an arbitrary nonzero estimate $\bsbB$, where $ r = r(\bsbB)$,   $ {\mathcal J} =\mathcal J(\bsbB)$ and $J=J(\bsbB)$.

\begin{enumerate}[a)]
\item  Let $\bsbS( \bsbB)=\bsbI[,  {\mathcal J}]$, the  submatrix formed by the columns  indexed by $\mathcal{J}$ in a $p\times p$ identity matrix.
\item If $r< J \wedge m$, find  $\bsbU( \bsbB) \in \mathbb O^{{J} \times  r}$   spanning the  column space   of $\bsbB[ {\mathcal J},]$ (which can be obtained by SVD or {QR} decomposition). Otherwise, set  $\bsbU( \bsbB) =\bsbI$, the  identity matrix of size $J\times J$.
\item The  \textbf{structural pattern} is given by the orthogonal matrix $\bsbS( \bsbB) \bsbU( \bsbB)$.
\end{enumerate}

According to the procedure, when $\bsbB$ has row sparsity, the index set of  nonzero rows,    ${\mathcal J}$,  determines the sparsity pattern.
If  $ \bsbB[ {\mathcal J},]$ has further rank deficiency, its  range  provides   the projection pattern.   This can be seen by  writing $\bsbX  \bsbB$ as
\begin{align*}
\bsbX  \bsbB & = \bsbX \bsbI[,\mathcal J] \bsbB[\mathcal J,]= \bsbX \bsbS (\bsbB)( \bsbU(\bsbB) \bsbU(\bsbB)^T \bsbB[\mathcal J,])\\
&=  \{\bsbX\bsbS (\bsbB) \bsbU(\bsbB)\} \cdot \{\bsbU(\bsbB)^T \bsbS(\bsbB)^T \bsbB\}.
\end{align*}
By construction,  all structural patterns $\bsbS(\bsbB)\bsbU(\bsbB)$  are orthogonal.

In our jointly sparse model,         $\bsbX\bsbS (\bsbB) \bsbU(\bsbB)$ is a new design matrix   that has only $ r$ factors   constructed from  $ J$ raw predictors.
Explicitly extracting  the factors      removes  estimation bias in regularized  estimation.   Given $\mathcal J$ and $r$, let $\bsbB_{ub}$ denote  the unbiased estimate that  solves \begin{align}\min_{\bsbB} \| \bsbY - \bsbX \bsbB\|_F^2 \mbox{ s.t. }  r(\bsbB) \leq r, \bsbB[\mathcal J^c, ]=\bsb{0}.\label{restrprob}\end{align} Assume   $\bsbB_{ub}\ne \bsb0$ and let   $\bsbS\bsbU$ be the associated structural pattern of $\bsbB_{ub}$. Then it is not difficult to prove   that   $ \bsbB_{ub}$ can always be decomposed into $\bsbS  \bsbU \cdot \bsbC$,  with $\bsbC$     the ordinary least squares (OLS) estimate from regressing $ \bsbY$ on $ \bsbX \bsbS \bsbU $.

 The procedure of cross-validating the structural patterns (or   \emph{s}-patterns, for short) that are associated with each estimate then follows, which is termed
\textit{structural cross-validation} (SCV).
 \begin{enumerate}
\item Split the row indices  of $(\bsbY, \bsbX)$ into $ K$ subsets. Denote the $k$th ($1\leq k\leq  K$) one of the data by $(\bsbY^{k}, \bsbX^{k})$, and the
remaining   by $(\bsbY^{-k}, \bsbX^{-k})$.
\item For each \emph{s}-pattern $\bsbS\bsbU$,  compute $\hat \bsbC^{-k}$, which is the OLS estimate from regressing $\bsbY^{-k}$ on $\bsbX^{-k} \bsbS  \bsbU$.   Let $\hat \bsbB^{-k}  = \bsbS \bsbU \hat \bsbC^{-k}$. Evaluate the prediction error of $\hat \bsbB^{-k}$  on $(\bsbY^{k}, \bsbX^{k})$. Repeat the process for all $k: 1\le k\le K$.
\item Summarize the  total validation error for each candidate model.
\end{enumerate}

Clearly, the $K$ models in SCV are   comparable because a  pattern  only acts on the  columns of $\bsbX$ and remains consistent across the (re-sampled) rows.
The fitting step    is cost-effective---indeed, each training  fits only a low-dimensional OLS
model  restricted to the  selected and/or projected dimensions. This  amounts to   sparse low-rank estimation with bias calibration in view of    \eqref{restrprob}.   In comparison,  ordinary   $K$-fold cross-validation  runs a more involved learning algorithm for  $K$ times. In structural cross-validation, the learning algorithm  is   called only once on  the overall data to generate candidate patterns.

It is worth mentioning that SCV is applicable to any   sparsity-inducing penalty, including those that are nonconvex.  It also applies to sole variable selection or rank selection and can be extended to non-Gaussian models such as generalized linear models; see \cite{SheGLMTISP}.

\section{In-sample Optimal Complexity Penalty  }
\label{sec:insample}
 Another problem  is if  $K$-fold SCV is rate optimal or if it applies a sufficiently large   penalty for overly complex models. 
Some popular choices are   $K=5$, $10$, or $n$. 
Under fixed $p$ and $n\rightarrow\infty$, some asymptotic conclusions have been obtained \citep{shao1997asymptotic,yang2005can}, but it is difficult to determine how large $n$ should be to apply the results in practice. Noticing that the problem  is similar to picking an appropriate complexity-penalty in information criteria, this section       studies a given model's   \textit{in-sample error}  \citep{ESL2}.
In-sample error is a form of prediction error  {conditioned} on the design matrix  \citep{ESL2}.
Our goal is to find
 a  \textbf{non-asymptotic} information criterion to   achieve  the optimal error rate in prediction \textit{without}    placing any    signal-to-noise    or design incoherence restriction or  requiring   the true model to be among the candidate models.  The predictive learning principle allows the data to inform the best parsimonious model. 


Assume    $\bsbY=\bsbX \bsbB^*+\bsbE$ with   $\EE[\bsbE \bsbE^T] = m\sigma^2 \bsbI$. Given an    arbitrary \emph{s}-pattern   $\bsbS \bsbU$  with   $\bar r = r(\bsbS \bsbU)$,   the associated restricted  estimate   is     $\hat \bsbB  = \bsbS \bsbU (\bsbX_{  \bsbS\bsbU }^T \bsbX_{ \bsbS\bsbU} )^{-1} \bsbX_{ \bsbS\bsbU}^T \bsbY$, where  we write $\bsbX \bsbS \bsbU$ as $\bsbX_{\bsbS \bsbU}$ for simplicity.
The {in-sample} test error for an  independent copy $\bsbY'$ of $\bsbY$ can be computed as follows 
\begin{align}
&\EE [\| \bsbY' - \bsbX \hat \bsbB\|_F^2 ] \notag \\
 =\, & nm\sigma^2 + \EE[\|\bsbX\hat \bsbB - \bsbX \bsbB^*\|_F^2] \notag\\
=\,&  n m \sigma^2 + \EE[\|\Proj_{\bsbX_{\bsbS\bsbU}}\bsbY - \Proj_{\bsbX_{\bsbS\bsbU}}\bsbX \bsbB^*\|_F^2]]  + \EE[\|(\bsbI - \Proj_{\bsbX_{\bsbS\bsbU}}) \bsbX  \bsbB^* \|_F^2] \notag\\
=\,&\EE [(nm + \bar r m) \sigma^2 + \EE[\|(\bsbI - \Proj_{\bsbX_{\bsbS\bsbU}}) \bsbX  \bsbB^* \|_F^2] \notag\\
=\, &\EE [ \|\bsbY - \bsbX \hat \bsbB\|_F^2  + 2  \sigma^2  \bar r m].  \label{insampleErr}
\end{align}
It appears that the criterion
\begin{align}
 \|\bsbY - \bsbX \hat \bsbB\|_F^2  + 2  \sigma^2  \bar r m \label{hyperIC}
\end{align}
 could be used for model comparison. When $m=1$, $\bar r = J$, \eqref{hyperIC}    is the  AIC (but     not when $m>1$). Nevertheless, the above  derivation   only shows the unbiasedness of   \eqref{hyperIC}   for a \textit{given} model.  Although  we are searching in a $p$-dimensional   space of all possible models, the model complexity penalty  $2  \sigma^2 \bar r m$     does not involve $p$.  In fact, minimizing~\eqref{hyperIC} will not lead to   an estimator with the lowest prediction error in sparse scenarios. 

To motivate the correct   complexity penalty, we study   the minimax   error rate for the class of models having both row sparsity and low rank.
 Given $J$ and $r$, define a signal class
\begin{align}
\begin{split}
S(J, r) = \{\bsbB \in \mathbb R^{p\times m}:   J(\bsbB)\leq J, r(\bsbB)\leq r \}. 
\end{split}\label{sigcls}
\end{align}
Any model in \eqref{sigcls} is jointly parsimonious in its rank and row support. Let $l(\cdot)$ be an arbitrary nondecreasing loss function with $l(0)=0, l\not\equiv 0$. Let $q=r(\bsbX)$. Define
\begin{align}
P_o(\bsbB) =   [q\wedge  J(\bsbB) +m - r(\bsbB)] r(\bsbB) + J(\bsbB)\log (e p / J(\bsbB) ). \label{optpen}
\end{align}
 We also write $P_o(J, r)$ because $P_o$ depends on $\bsbB$ through $J(\bsbB)$ and $r(\bsbB)$.
\begin{theorem} \label{th:minimax}
Let $\bsbY = \bsbX \bsbB^* + \bsbE$ where $\bsbE=[e_{ik}]$ with $e_{ik}$ $\overset{ i.i.d. }\sim \mathcal N(0, \sigma^2)$, and    $\bsbB^*\in S(J, r)$ for some  $J, r$  satisfying   $1\le J \le p/2$, $p\ge 2$, $r(q\wedge  J+m -r )\ge 16$, $1\le r \le \min\{(J+m)/2,  2J\}$.  Assume  the following  restricted condition-number condition: (i) when $J\le q$,  there exist $\underline{\kappa}, \overline{\kappa} > 0$ such that  $
\underline{\kappa} \|\bsbB\|_F^2\leq \|\bsbX \bsbB\|_F^2 \leq \overline{\kappa}  \| \bsbB\|_F^2
$
for any $\bsbB: J(\bsbB)\le J, r(\bsbB)\le r$ and   $
\underline{\kappa} / \overline{\kappa}$ is a positive constant; (ii) when $J>q$,  $
 \sigma_{\min}^2(\bsbX) /\sigma_{\max}^2(\bsbX)$ is a positive constant, where $\sigma_{\min}(\bsbX)$ and $\sigma_{\max}(\bsbX)$ are the smallest and the largest {\emph{nonzero}} singular values of $\bsbX$, respectively. Then  there exist positive constants $c, c'$,  depending on $l(\cdot)$ only, such that when $J \le q$,
\begin{align}
\inf_{\hat \bsbB} \sup_{\bsbB^* \in S(J, r)} \EE[l(\|\bsbX \bsbB^* - \bsbX \hat\bsbB\|_F^2/(c \sigma^2 P_o(J, r)))] \geq c' >0,
\end{align}
and when $J>q$, $
\inf_{\hat \bsbB} \sup_{\bsbB^* \in S(J, r)} \EE[l(\|\bsbX \bsbB^* - \bsbX \hat\bsbB\|_F^2/(c \sigma^2 (q  +m - r ) r ))] \geq c' >0.
$
\end{theorem}

Notice that the restricted condition-number condition takes a more relaxed form when $J$ is large ($J>q$) and allows $p$   to be greater than $n$. Setting $l(t)=t$, the   theorem shows that    the minimax optimal  rate for the risk $\EE \|\bsbX \bsbB^* - \bsbX \hat\bsbB\|_F^2$ is $\sigma^2 P_o(J, r) $. Also, with $l(t)=1_{t\geq 1}$, for any estimator $\hat\bsbB$, there exists a $\bsbB^*\in S(J, r)$ such that $\| \bsbX \bsbB^* - \bsbX \hat\bsbB\|_F^2\gtrsim \sigma^2 P_o(J, r)$ with positive probability. The minimax lower bound 
 suggests   $P_o$ be an ideal model-complexity penalty, which leads to a new information criterion for model comparison.

\begin{theorem} \label{th:pic}
Assume  $\bsbY = \bsbX \bsbB^* + \bsbE$ where $\vect(\bsbE)$ is sub-Gaussian with mean $\bsb0$ and scale bounded by $\sigma$. For any collection of (possibly random) non-zero matrices $\bsbB_1$, $\bsbB_2$, $\ldots$, select
the optimal one  by $$
\hat \bsbB = {\mathop{\rm arg\, min}}_{\bsbB_l }\frac{1}{2}\| \bsbY - \bsbX \bsbB_l\|^2_F + A \sigma^2 P_o(\bsbB_l),
$$
with the constant $A$ appropriately large.  Then $\hat \bsbB $
satisfies
\begin{align}
\EE[ \|  \bsbX \hat \bsbB  - \bsbX \bsbB^* \|^2_F ]
\lesssim \inf_{l}      \EE [ \|  \bsbX \bsbB_l - \bsbX \bsbB^* \|_F^2 +\sigma^2  P_o(\bsbB_l) ].\label{picrateres}
\end{align}
\end{theorem}

We refer to this new information criterion as the   \textit{Predictive Information Criterion} (\textbf{PIC}). \eqref{picrateres} does not require $ \bsbB_l $ ($l\ge 1$) to cover the true model. But when $\bsbB^*$ is among the candidate matrices or there exists some $\bsbB_l$ close to $\bsbB^*$,   the risk of $\hat \bsbB$ is bounded above  by $\sigma^2 P_o \left( \bsbB^* \right)$ up to multiplicative constants.   When   $m=1$,   the complexity penalty in PIC is of the   order $\sigma^2 J \log ( ep/ J )$, is similar to (but finer than) the rate in RIC   \citep{FosterGeorge}. However, we do not   assume that $n - p$ must be large as in the derivation of RIC.
In fact,   PIC is     non-asymptotic in nature,  and   requires no large-$n$ assumption. Moreover, it does \emph{not} require
any incoherence condition which is commonly assumed in the literature.
From the bias-variance tradeoff in \eqref{picrateres}, setting $\bsbB_l = \bsbB^*$, the noise-free truth, may not yield the most accurate and parsimonious model,
especially when the noise level is large. This is well known in wavelet studies \citep{donoho1994}. When the signal is properly large, minimizing PIC also recovers the row support of $\bsbB^*$ with high probability under some regularity conditions; see Theorem~\ref{th:supp} in the  Appendix. 


In \eqref{optpen}, $[J(\bsbB) +m - r(\bsbB)] r(\bsbB)$ corresponds to the number of free parameters of the model or   degrees of freedom (\textbf{DF}),   typically  less than $mn$, while $J(\bsbB)\log (e p / J(\bsbB))$ characterizes the inflation (\textbf{IF}) due to the selection among $p$ predictors. The model-complexity penalty can then be written as $A_1 \sigma^2  \mbox{\textbf{DF}} + A_2 \sigma^2 \mbox{\textbf{IF}}$, where $A_1, A_2$ are constants. The additive form 
 is in contrast to the multiplicative form $ c(n, p) \times \sigma^2\mbox{\textbf{DF}} $ that is widely used in the literature, for example, $c(n, p) = \log p$ in BIC.


The unknown noise scale remains an issue in supervised learning, because in large-$p$ settings  estimating $\sigma$ could be as challenging as estimating the mean.
The following theorem, as a generalization of Theorem 3 in \cite{she2017selective},  presents four scale-free forms of PIC to \textit{bypass} the scale estimation under a model sparsity  assumption.
\begin{theorem} \label{th:sf-pic}
Let $\bsbY = \bsbX \bsbB^* + \bsbE$ where $\bsbE=[e_{ik}]$ with $e_{ik}$ $\overset{ i.i.d. }\sim \mathcal N(0, \sigma^2)$. Suppose the true model is parsimonious in the sense that $P_o(\bsbB^*) \allowbreak < mn / A_0$ for some  constant $A_0>0$.
  Let $\delta(\bsbB)=A P_o(\bsbB)/(mn )$ for some constant $ A<A_0$, and so $\delta(\bsbB^*)<1$. Consider the following  criteria
\begin{align}
& {\|\bsbY - \bsbX \bsbB\|_F^2}/[{ 1 - \delta(\bsbB)}],\label{eq:sf-pic}\\
& {   \|\bsbY - \bsbX \bsbB\|_F^2}/{ [1 - \delta(\bsbB) ]^2},\label{gcv-pic}\\
 & \log \{\|\bsbY - \bsbX \bsbB\|_F^2\}+   \delta(\bsbB),\label{log-pic}\\
 & \|\bsbY - \bsbX \bsbB\|_F^2  +    \delta(\bsbB){\|\bsbY - \bsbX \bsbB\|_F^2}.    \label{plugin-pic}
\end{align}
 Then, for sufficiently large  values of $A_0,A$,   any $\hat\bsbB$ that minimizes \eqref{eq:sf-pic}, \eqref{gcv-pic}, \eqref{log-pic}, or \eqref{plugin-pic}     subject to $\delta(\bsbB)<1$    satisfies $\| \bsbX \hat \bsbB - \bsbX \bsbB^*\|_F^2\lesssim \sigma^2 P_o(\bsbB^*)$  with probability at least $1 -C p^{-c}-C'\exp(-c' m n )$ for some  constants $C, C', c, c'>0$. 
\end{theorem}

We   call \eqref{eq:sf-pic}, \eqref{gcv-pic}, \eqref{log-pic}, \eqref{plugin-pic} the  fractional form, GCV  form,   logarithmic form, and   plug-in form of PIC, respectively. With the inflation term removed,    \eqref{gcv-pic} shares   similarity with GCV~\citep{wahba1990spline}; the log form  is commonly seen when   applying AIC or BIC with an unknown $\sigma^2$;  the penalty in  \eqref{plugin-pic} can be written as $A \hat \sigma ^2 P_o(\bsbB)$ with $\hat\sigma^2 =  \|\bsbY - \bsbX \bsbB\|_F^2/(mn),$ which resembles Mallows' $C_p$~\citep{mallows1973some}. The sparsity assumption and constraint cannot be dropped,  which in turn rule out over-complex models. 

 In common with most   non-asymptotic analyses, we showed the optimal rate but not the optimal numerical constants. The absolute constants    can be  determined by Monte Carlo experiments---for example, in the fractional form  
  we recommend using $\|\bsbY - \bsbX \bsbB\|_F^2/[ 1 - (2 \cdot \mbox{DF} + 1.8 \cdot \mbox{IF})/(mn)] $. 


\section{Rate Calibration of Cross-validation}
This  section studies    \textit{extra-sample error} that cross-validation methods try to estimate \citep[Chapter 7]{ESL2}.
 It turns out that there is a connection between the in-sample error and the extra-sample error, which can be used to guide cross-validation.

To define the extra-sample error, we assume that   the row observations in the design are i.i.d.  and independent of the noise component.
Specifically, assume $\bsbY = \bsbX \bsbB^* + \bsbE$, where $\bsbE=[e_{ik}]$ satisfies $\EE [\bsbE \bsbE^T]=\sigma^2 m \bsbI$, $\bsbX=[\tilde \bsbx_1 \, \cdots \, \tilde \bsbx_n]^T$ is  independent of $\bsbE$ and has  i.i.d. rows   with $\bsbSig$  (positive-definite) as the  covariance matrix. The $s$-pattern from $\bsbB^{*}$ is denoted by $\bsbS^{*} \bsbU^{*}$. Let $\bsbS \bsbU$ be a given candidate   \emph{s}-pattern with $\bar r  =r(\bsbS( \bsbB) \bsbU( \bsbB)) $.  Recall the training error and the structural cross-validation error
\begin{align*}
& \mbox{Trn-Err}     \triangleq  \|\bsbY - \bsbX \hat\bsbB \|_F^2,  \quad
\\ &\mbox{CV-Err}     \triangleq   \sum_{k=1}^K \|\bsbY^{{k}} - \bsbX^{{k}} \hat\bsbB^{-{k}} \|_F^2,
\end{align*}
 where $\hat\bsbB$,  $\hat \bsbB^{-{k}}$ are the restricted OLS estimates associated with the given \emph{s}-pattern, obtained  on   the overall data and the data without the $k$-th subset, respectively. Concretely,   $$\hat \bsbB  = \bsbS \bsbU (\bsbX_{  \bsbS\bsbU }^T \bsbX_{ \bsbS\bsbU} )^{-1} \bsbX_{ \bsbS\bsbU}^T \bsbY,  \hat  \bsbB ^{-k}  = \bsbS \bsbU ((\bsbX_{  \bsbS\bsbU }^{-k}) ^T \bsbX_{ \bsbS\bsbU}^{-k} )^{-1}(\bsbX_{  \bsbS\bsbU }^{-k}) ^T\bsbY^{-k}$$  with $\bsbX_{\bsbS\bsbU}=\bsbX \bsbS \bsbU$ and $\bsbX_{  \bsbS\bsbU }^{-k}=\bsbX^{-k}   \bsbS\bsbU $, where  all   the restricted OLS  problems are assumed to be non-degenerate.
Suppose that    $n=dK$ for some integer $d$ (and so  $d= n/K $). The theorem below  gives an identity of the cross-validation error. 
\begin{theorem} \label{th:iden}
Given any \emph{s}-pattern $\bsbS\bsbU$ with      $\bar r = r(\bsbS\bsbU)$, the following {identity} holds
\begin{align}
\EE[\mbox{CV-Err}] &= \EE[\mbox{Trn-Err}] + D + U,
\end{align}
where
\begin{align}
D  = \  & m \bar r \sigma^2(1 + \frac{n}{\bar r}{ \EE[Tr\{ (\bsbZ_{n-d}^T \bsbZ_{n-d})^{-1}\}] } ),
\end{align}
and
\begin{align}
\begin{split}
U \  = \ \  &\EE[\|\bsbX_{\bsbS\bsbU} ( { \bar  \bsbB_{\bsbS\bsbU, \bsbS^* \bsbU^*}} - \EE {   \bsbB_{\bsbS\bsbU, \bsbS^* \bsbU^*}} )\|_F^2] \\ \ \     +& \EE[\|\bsbSig^{1/2}_  {\bsbS\bsbU} ( {    \bsbB_{\bsbS\bsbU, \bsbS^* \bsbU^*}} - \EE {   \bsbB_{\bsbS\bsbU, \bsbS^* \bsbU^*}} )\|_F^2].\end{split}\label{defU}
\end{align}
Here, $\bsbZ_{n-d}\in \mathbb R^{(n-d)\times \bar r}$ is a submatrix composed of (any) $n-d$ rows of $\bsbZ= \bsbX_{  \bsbS \bsbU}  (\bsbS \bsbU)^T \allowbreak\bsbSig^{-1/2} \bsbS \bsbU $, $ \bsbSig_{\bsbS\bsbU} = (\bsbS\bsbU)^T (n\bsbSig) \bsbS \bsbU$,
$  \bsbB_{\bsbS\bsbU, \bsbS^* \bsbU^*}$ is $((\bsbX_{\bsbS\bsbU}^{-k})^T \bsbX_{\bsbS\bsbU}^{-k})^{-1} \allowbreak (\bsbX_{\bsbS\bsbU}^{-k})^T (\bsbX^{-k}  \Proj^{\circ}   \bsbB^*)$ with say $k=1$,  $ \bar  \bsbB_{\bsbS\bsbU, \bsbS^* \bsbU^*}= ( \bsbX_{\bsbS\bsbU}^T \bsbX_{\bsbS\bsbU} )^{-1}  \bsbX_{\bsbS\bsbU}^T \allowbreak (\bsbX  \Proj^{\circ}   \bsbB^*)$, and $\Proj^\circ$ is the orthogonal projection onto $\Proj_{\bsbS^* \bsbU^*}\setminus (\Proj_{\bsbS^* \bsbU^*}\cap\Proj_{\bsbS \bsbU})$, i.e., the orthogonal complement of $\Proj_{\bsbS^* \bsbU^*}\cap\Proj_{\bsbS \bsbU}$ in the subspace $\Proj_{\bsbS^* \bsbU^*}$.  
\end{theorem}

According to \eqref{defU},      $U$ is always nonnegative. In the   classical asymptotic regime where the sample size tends to $\infty$ and    $m, p$ are fixed, $\bsbB_{\bsbS\bsbU, \bsbS^* \bsbU^*}$, as well as  $\bar \bsbB_{\bsbS\bsbU, \bsbS^* \bsbU^*}$, will approach $\EE {   \bsbB_{\bsbS\bsbU, \bsbS^* \bsbU^*}}$ because of the law of large numbers. In finite samples, however, $U$ incurs extra   cost for
underfitting models,  which could be   an advantage over some information criteria. From the definition of   $\Proj^\circ$, $U$   vanishes for over-fitting
models that   satisfy  $\bsbB^*\in \Proj_{\bsbS\bsbU}$, and hence is   not active  for ``large" models.

  The term  $D$ penalizes the cardinality of the model, and relates to  the
degrees of freedom. The definition of  $D$ is based on $\bsbZ$  which is   just    $\bsbX {\bsbS \bsbU}$    decorrelated. Because the rows of $\bsbZ$ are isotropic (i.e., the covariance of each row vector is $\bsbI$),  it is not difficult to show an upper bound of $D$  using  random matrix theory. We give a   corollary   as an illustration under the assumption that the rows of $\bsbX$  are i.i.d.   $\mathcal N (\boldsymbol 0, \bsbSig)$ (such designs are  widely used in simulation studies).
\begin{corollary} \label{iden-cor} Under the setup of the previous theorem and  the Gaussian assumption of the design,
\begin{align}
\EE[\mbox{CV-Err}] = & \EE \left [ \mbox{Trn-Err}  +{ \frac{(2-(\bar r +1)/n)K-1}{(1-(\bar r +1)/n)K-1} m\bar r \sigma^2 }\right ] + U, \label{guassian-cverr}
\end{align}
where $U$ is as defined in Theorem \ref{th:iden}. In particular, for ordinary    variable selection   with  $m=1$, the identity for any given support $\mathcal J$ reads
\begin{align*}
 \EE[\mbox{CV-Err}] &=     \EE  [ \mbox{Trn-Err}  ]+{ \frac{(2-(J +1)/n)K-1}{(1-(J +1)/n)K-1} J \sigma^2 }   \\&  \ +\EE[ \|\bsbX_{\mathcal J} ( \bar \bsbb_{\mathcal J, \mathcal J^* } - \EE    \bsbb_{\mathcal J, \mathcal J^* } )\|_2^2] +\EE[ \|\bsbSig_{\mathcal J} ^{1/2}(   \bsbb_{\mathcal J, \mathcal J^* } - \EE    \bsbb_{\mathcal J, \mathcal J^* } )\|_2^2] ,
\end{align*}
where  $J=|\mathcal J|$,  $\bsbSig_{\mathcal J} = n\bsbSig[\mathcal J, \mathcal J]$, $\bsbX_{\mathcal J} = \bsbX[\mathcal J,]$, $ \bar \bsbb_{\mathcal J, \mathcal J^*}= (\bsbX_{\mathcal J}^T \bsbX_{\mathcal J})^{-1} \bsbX_{\mathcal J}^T \bsbX_{\mathcal J^* \cap \mathcal J^c} \allowbreak\bsbb_{{ \mathcal J^* \cap \mathcal J^c }}^*$ and $  \bsbb_{\mathcal J, \mathcal J^*}= ((\bsbX_{\mathcal J}^{-k})^T \allowbreak \bsbX_{\mathcal J}^{-k})^{-1} (\bsbX_{\mathcal J}^{-k})^T \bsbX_{\mathcal J^* \cap \mathcal J^c}^{-k} \bsbb_{{ \mathcal J^* \cap \mathcal J^c }}^*$.

\end{corollary}

In  \eqref{guassian-cverr},
 $D$ is a decreasing function in $K$. When $K=n$, $D$ is  essentially   $2 m\bar r \sigma^2$; even when    $K$  is as  small as $2$,    $D\le \{3+4\theta/(1-2\theta)\} m \bar r\sigma^2$ under $(\bar r+1)/n\le \theta$, and so for $n$   large, the   
    $\mbox{CV-Err}$ of an over-fitting model is   no larger than   \begin{align}\mbox{Trn-err} + A m \bar r \sigma^2. \label{plaincvrate}\end{align}
This means that   the extra-sample error of $K$-fold cross-validation is not significantly different from that in the in-sample error formula \eqref{insampleErr}, which is, however, perhaps natural.

Returning to the jointly sparse model (where $\bar r = r(\bsbS(\bsbB)\bsbU(\bsbB))=r(\bsbB)=r$),
    combining the identity  with PIC gives a choice   of   $K$   to  match  the minimax rate in~\eqref{optpen}:
$
K=\{A P_o(J, r)  - mr\}/\{A P_o(J, r)    -2mr - (A P_o(J, r)  -mr)(r+1)/n\}$.
In large-$p$ problems, however, the value is   below $2$ and  unattainable; consequently,   $K$-fold cross-validation ($K\ge 2$) cannot   penalize the  complexity of the     model sufficiently heavily. 

 One possible fix is to use   delete-$d$ cross-validation~\citep{shao1993linear} which  removes $d$ observations in each training and  $d$ can be larger than $n/2$. A similar identity to that in Theorem \ref{th:iden}  holds for this form of  cross-validation,  and  to match the PIC rate, the training sample size    $n-d$ should be  of order   $ {mnr}/( A P_o(J,r)  -mr)+r+1 $ or $   O( n\cdot  {mr}/({  J r + J\log p})) $. Delete-$d$ cross-validation can be implemented  by enumerating all   subsets of size $d$, or in a stochastic fashion by randomly splitting the whole dataset many times. But neither is   computationally efficient for large $n$ and we shall not pursue further in this work.

By contrast, the commonly used 5-fold CV and 10-fold CV are much less expensive, and these  fold choices enjoy small variance \citep{ESL2}.
 Although $K$-fold CV ($2\le K\le n$) is suboptimal as seen from \eqref{plaincvrate} and \eqref{optpen}, we can make a   rate calibration  as in the in-sample error case. Let $\mbox{R} = (q\wedge  J   - r ) r$ and recall $\mbox{IF} = J \log (ep/J)$. From Theorems~\ref{th:pic} and~\ref{th:iden}, if $\sigma^2$ is known, one can append a bias correction term of the order $\sigma^2 \mbox{R} + \sigma^2 \mathrm{IF}$ to the CV-Err. When $\sigma$ is unknown, motivated by the plug-in form of the $\sigma$-free PIC (cf. Theorem~\ref{th:sf-pic}),   we can use the following   calibrated  structural cross-validation error as the model selection criterion:
\begin{align}
\mbox{\textbf{SCV-Err}}: =    {\mbox{CV-Err}} + \alpha_1 (\mbox{Trn-Err}/( mn))  \mbox{R}   +\alpha_2   (\mbox{Trn-Err}/( mn) ) \mbox{IF}, \label{ratecorrSCV}
\end{align}
where   $\alpha_1, \alpha_2$ are constants.
The rate correction is evident from the identity: the calibrated CV error \eqref{ratecorrSCV} has a complexity penalty of the form  $\sigma^2   \mbox{{DF}} + \sigma^2\mbox{{IF}}$, complying with the studies in Section \ref{sec:insample}. According to  the constraint $\delta(\bsbB) < 1$ in  Theorem~\ref{th:sf-pic},         the candidate models that are far too complex (say $\alpha_1 \mbox{DF} + \alpha_2 \mbox{IF} > mn$) should be excluded;   equivalently,  we set their $\mbox{SCV-Err}$ to be $+\infty$. For $K=5$, we recommend  $\alpha_1=4.6$ and $\alpha_2=3.5$ on the basis of  Monte Carlo experiments.
 (The rate correction of SCV applies to any choice of $K$, but the  numerical constants may be different.)   Of course, other forms are possible---for example,   
  the fractional form, $  \mbox{CV-Err} /[ 1 - \alpha_1 \mbox{R}  / (mn + \alpha mr) -\alpha_2    \mbox{IF}/  (mn+\alpha mr)  ]$, or simply $$ {\mbox{CV-Err}}/( 1 - \alpha_1 \mbox{R}  /( mn)  -  \alpha_2    \mbox{IF}/(  mn)  )$$ with $\alpha_1=2$ and $\alpha_2=2.4$  also gives satisfactory performance.
We   apply the plug-in form \eqref{ratecorrSCV}  as the SCV error in the regression setting.

\section{Experiments}
 \subsection{Simulations}
We use the following setup for the simulation studies. The design matrix $\bsbX$ has i.i.d. rows from $\mathcal{N} \left( \boldsymbol 0, \bsbSig \right)$ with $\Sigma_{jk} = \rho^{ |j - k|}$, $\rho >0$, $1 \leq j$, $k \leq p$. The coefficient matrix has the form $\bsbB^{*} = [ b (\bsbA_0 \bsbA_1)^T \,\, \boldsymbol 0 ]^T$
where $b$ is a constant, $\bsbA_0$ is a $J \times r$ matrix and $\bsbA_1$ is an $r \times m$ matrix. Entries in $\bsbA_0$ and $\bsbA_1$ are i.i.d. standard Gaussian. Entries in $\bsbB^{*}$ past the $J$th row are all zero. The matrix $\bsbE$ has i.i.d. standard Gaussian entries and the response matrix is $\bsbY = \bsbX \bsbB^{*} + \sigma \bsbE$.

We consider the following cases of $p > n$ and $n >p$, with two different correlation levels and signal strengths in generating the design $\bsbX$:
\begin{enumerate}
\item $n > p$: $n=100$, $J =30$, $p=60$, $m=15$, $r=5$, $\sigma=1$, $\rho=0.1, 0.5$, $b = 0.1, 0.5$.
\item $p > n$: $n = 30$, $J = 15$, $p = 100$, $m = 10$, $r=2$, $\sigma=1$, $\rho=0.1, 0.5$, $b = 0.2, 1$.
\end{enumerate}
We ran $200$ simulations for each setting. In common with other studies on high-dimensional parameter tuning~\citep{chenchen}, we called a learning algorithm (cf. \cite{BSW12} and~\cite{ChenHuang12}) on each synthetic dataset to compute a solution path of candidate estimates. Then the associated projection-selection patterns were extracted according to the procedure in Section~\ref{sec:structural}.

Six model selection criteria were compared: AIC, BIC, EBIC, all taking    the logarithmic  form due to the unknown noise level, PIC as in  \eqref{eq:sf-pic}, $K$-fold CV with $K=2,5,10$, and $5$-fold SCV---see~\eqref{ratecorrSCV}. (
We found that the performance of $K$-fold SCV was similar for $K=2,5,10$, and so  focused on $5$-fold SCV.) The prediction accuracy was evaluated     by the   mean squared   error (MSE): $\EE[ Tr \{ (\hat \bsbB - \bsbB^*)^T \bsbSig  (\hat \bsbB - \bsbB^*) \}]/m$.  The median of the MSEs over all simulations was then computed to
represent the goodness of fit of a model. Selection performance was assessed by the median number of predictors ($\hat J$) and the median rank estimate ($\hat r$) over all simulations. The rates of left-out noise-free true variables (M for missing) and incorrectly included variables (FA for false alarms) averaged across all runs are also reported. When the signal strength is high, a successful variable selection method minimizes the M- and FA-rates with a preference for a low miss rate since it is undesirable to leave out true features. In contrast, when the noise level is very high, there is no reason to believe that the noise-free \textit{simulation} truth yields the best predictive model from the observed data. Therefore, in the low SNR situations our ultimate concern is the prediction error. The results for the six model comparison criteria are presented in Tables~\ref{tab:nLarge_bSmall}--\ref{tab:pLarge_bLarge} and are summarized below.

The PIC and SCV methods have superior prediction performance across nearly all combinations of signal strength and correlation considered here. EBIC's prediction in the lower signal-to-noise ratio (SNR) experiments (see, for example, Table~\ref{tab:nLarge_bSmall}) is evidence that its comparatively high complexity penalty induces too much regularization for these particular setups. In the larger SNR experiments (cf. Tables~\ref{tab:nLarge_bLarge} and~\ref{tab:pLarge_bLarge}), the EBIC's prediction performance is comparable to PIC and SCV. In the $p > n$ cases, AIC has the highest MSEs among all the information criteria, which is unsurprising because it is well known that AIC underpenalizes for large $p$.

CV often has the highest MSEs among all methods, indicating that a  rate correction is absolutely necessary. In all experiments the prediction performance of CV across $K=2,5,10$ is similar, verifying the discussion following Corollary~\ref{iden-cor}.   $K=5$ showed a modest improvement compared to $K=2,10$.

Interestingly, the PIC and SCV miss some of the  variables  specified in the absence of noise when the signal strength is weak (cf. Tables~\ref{tab:nLarge_bSmall} and~\ref{tab:pLarge_bSmall}), but this is perhaps natural because a larger amount of regularization may be required to achieve low prediction error. In these small SNR experiments, the AIC tends to miss the fewest of these noise-free   variables, but this comes at the cost of higher FA-rates and as stated before it often has poor prediction accuracy. In both weak and strong signal strength situations, the BIC and especially the EBIC   tend to have high missing rates compared to the AIC because of their larger penalty terms  but,   unlike the PIC and SCV, they are rarely able to select a model with parsimony \emph{and} low prediction error in both weak and strong SNR situations.

In the $n>p$ experiments with relatively stronger signal strength (Table~\ref{tab:nLarge_bLarge}), PIC and SCV's variable selection performance nearly equals or exceeds that of all other methods. In fact, their M- and FA- rates are nearly 0 in Table~\ref{tab:nLarge_bLarge}, suggesting that here the true variables are also highly predictive. Indeed, in such a large SNR situation (see Theorem~\ref{th:supp}) PIC can recover the noise-free row support with high probability.

Compared to the other methods, the median rank values of the PIC and SCV are most consistently equal to their true values. In nearly every experiment the AIC and CV overestimate $r$ while the BIC and especially EBIC underestimate the rank when the signal strength is relatively weak.

\begin{table}
    \caption{\label{tab:nLarge_bSmall} Performance comparisons between AIC, BIC, EBIC, PIC, $2$-fold CV, $10$-fold CV, $5$-fold CV and $5$-fold SCV in the $n > p$ experiment with smaller signal strength ($b=0.1$). MSEs are scaled for ease of comparison and M- and FA-rates are in percentages.}
    \centering
    \begin{tabular}{l c c c c c c c c c c}
    \hline
    & \multicolumn{5}{c}{$\rho = 0.1$} & \multicolumn{5}{c}{$\rho = 0.5$} \\
    \cmidrule(lr){2-6} \cmidrule{7-11}
    & \textbf{MSE} & $\boldsymbol{ \hat J }$ & $\boldsymbol{ \hat r }$ & \textbf{M}   & \textbf{FA} & \textbf{MSE} & $\boldsymbol{ \hat J }$ & $\boldsymbol{ \hat r }$ & \textbf{M}   & \textbf{FA} \\
    \hline
\textbf{AIC} & 43 & 37 & 6 & 4 & 29 & 44 & 38 & 6 & 7 & 35\\
\textbf{BIC} & 45 & 20 & 3 & 33 & 0 & 42 & 18 & 3 & 41 & 1\\
\textbf{EBIC} & 104 & 8 & 1 & 74 & 0 & 103 & 7 & 1 & 75 & 0\\
\hdashline
\textbf{PIC} & 26 & 29 & 4 & 11 & 7 & 28 & 27 & 4 & 18 & 9\\
\specialrule{.9pt}{.1pt}{.1pt}
2-\textbf{CV} & 99 & 50 & 10 & 1 & 68 & 99 & 50 & 10 & 3 & 69\\
10-\textbf{CV} & 99 & 50 & 10 & 1 & 68 & 99 & 50 & 10 & 3 & 69\\
5-\textbf{CV} & 96 & 50 & 10 & 1 & 68 &  96 & 50 & 10 & 3 & 69\\
\hdashline
5-\textbf{SCV}  & 28 & 29 & 5 & 10 & 9 & 30 & 27 & 5 & 19 & 9\\
\hline
    \end{tabular}
\end{table}

\begin{table}
    \caption{\label{tab:nLarge_bLarge} Performance comparisons between AIC, BIC, EBIC, PIC, $2$-fold CV, $10$-fold CV, $5$-fold CV and $5$-fold SCV in the $n > p$ experiment with larger signal strength ($b=0.5$). MSEs are scaled for ease of comparison and M- and FA-rates are in percentages.}
    \centering
    \begin{tabular}{l c c c c c c c c c c}
    \hline
    & \multicolumn{5}{c}{$\rho = 0.1$} & \multicolumn{5}{c}{$\rho = 0.5$} \\
    \cmidrule(lr){2-6} \cmidrule{7-11}
    & \textbf{MSE} & $\boldsymbol{ \hat J }$ & $\boldsymbol{ \hat r }$ & \textbf{M}   & \textbf{FA} & \textbf{MSE} & $\boldsymbol{ \hat J }$ & $\boldsymbol{ \hat r }$ & \textbf{M}   & \textbf{FA} \\
\hline
\textbf{AIC} & 13 & 34 & 6 & 0 & 17 & 12 & 33 & 6 & 0 & 17\\
\textbf{BIC} & 7 & 30 & 5 & 0 & 0 & 7 & 30 & 5 & 0 & 0\\
\textbf{EBIC} & 7 & 30 & 5 & 0 & 0 & 7 & 30 & 5 & 0 & 0\\
\hdashline
\textbf{PIC} & 7 & 30 & 5 & 0 & 1 & 7 & 30 & 5 & 0 & 0\\
\specialrule{.9pt}{.1pt}{.1pt}
2-\textbf{CV} & 40 & 50 & 10 & 0 & 67 & 40 & 50 & 10 & 0 & 67\\
10-\textbf{CV} & 40 & 50 & 10 & 0 & 67 & 40 & 50 & 10 & 0 & 67\\
5-\textbf{CV} & 40 & 50 & 10 & 0 & 66 & 40 & 50 & 10 & 0 & 66 \\
\hdashline
5-\textbf{SCV} & 7 & 30 & 5 & 0 & 2 & 7 & 30 & 5 & 0 & 1\\
\hline
    \end{tabular}
\end{table}

\begin{table}
    \caption{\label{tab:pLarge_bSmall} Performance comparisons between AIC, BIC, EBIC, PIC, $2$-fold CV, $10$-fold CV, $5$-fold CV and $5$-fold SCV in the $p >n$ experiment with smaller signal strength ($b=0.2$). MSEs are scaled for ease of comparison and M- and FA-rates are in percentages.}
    \centering
    \begin{tabular}{l c c c c c c c c c c}
    \hline
    & \multicolumn{5}{c}{$\rho = 0.1$} & \multicolumn{5}{c}{$\rho = 0.5$} \\
    \cmidrule(lr){2-6} \cmidrule{7-11}
    & \textbf{MSE} & $\boldsymbol{ \hat J }$ & $\boldsymbol{ \hat r }$ & \textbf{M}   & \textbf{FA} & \textbf{MSE} & $\boldsymbol{ \hat J }$ & $\boldsymbol{ \hat r }$ & \textbf{M}   & \textbf{FA} \\
\hline
\textbf{AIC} & 93 & 24 & 4 & 44 & 17 & 79 & 24 & 4 & 36 & 16\\
\textbf{BIC} & 42 & 12 & 1 & 59 & 7 & 31 & 10 & 2 & 56 & 5\\
\textbf{EBIC} & 43 & 7 & 1 & 71 & 3 & 36 & 4 & 1 & 74 & 1 \\
\hdashline
\textbf{PIC} & 40 & 10 & 2 & 62 & 5 & 28 & 9 & 2 & 59 & 3\\
\specialrule{.9pt}{.1pt}{.1pt}
2-\textbf{CV} & 122 & 25 & 4 & 43 & 19 & 97 & 25 & 4 & 35 & 18 \\
10-\textbf{CV} & 130 & 25 & 4 & 43 & 19 & 99 & 25 & 4 & 35 & 18 \\
5-\textbf{CV} & 128 & 25 & 4 & 43 & 19 & 99 & 25 & 4 & 35 & 18 \\
\hdashline
5-\textbf{SCV}  & 41 & 10 & 2 & 62 & 5  & 29 & 8 & 2 & 61 & 3 \\
\hline
    \end{tabular}
\end{table}

\begin{table}
   \caption{\label{tab:pLarge_bLarge} Performance comparisons between AIC, BIC, EBIC, PIC, $2$-fold CV, $10$-fold CV, $5$-fold CV and $5$-fold SCV in the $p >n$ experiment with larger signal strength ($b=1$). MSEs are scaled for ease of comparison and M- and FA-rates are in percentages.}
    \centering
    \begin{tabular}{l c c c c c c c c c c}
    \hline
    & \multicolumn{5}{c}{$\rho = 0.1$} & \multicolumn{5}{c}{$\rho = 0.5$} \\
    \cmidrule(lr){2-6} \cmidrule{7-11}
    & \textbf{MSE} & $\boldsymbol{ \hat J }$ & $\boldsymbol{ \hat r }$ & \textbf{M}   & \textbf{FA} & \textbf{MSE} & $\boldsymbol{ \hat J }$ & $\boldsymbol{ \hat r }$ & \textbf{M}   & \textbf{FA} \\
\hline
\textbf{AIC} & 47 & 25 & 2 & 40 & 18 & 31 & 25 & 2 & 32 & 16 \\
\textbf{BIC} & 43 & 24 & 2 & 40 & 17 & 30 & 24 & 2 & 33 & 15 \\
\textbf{EBIC} & 34 & 21 & 2 & 43 & 12 & 26 & 19 & 2 & 38 & 11 \\
\hdashline
\textbf{PIC} & 41 & 23 & 2 & 40 & 15 & 28 & 21 & 2 & 35 & 13 \\
\specialrule{.9pt}{.1pt}{.1pt}
2-\textbf{CV} & 47 & 25 & 2 & 40 & 18 & 34 & 25 & 3 & 31 & 17 \\
10-\textbf{CV} & 48 & 25 & 3 & 40 & 18 & 34 & 25 & 3 & 31 & 17 \\
5-\textbf{CV} & 50 & 25 & 2 & 39 & 18  & 35 & 25 & 4 & 31 & 17 \\
\hdashline
5-\textbf{SCV} & 34 & 11 & 2 & 58 & 5 & 22 & 11 & 2 & 53 & 4 \\
\hline

    \end{tabular}
\end{table}

The prediction performance of the PIC and SCV was similar in the $n>p$ experiments (cf. Tables~\ref{tab:nLarge_bSmall} and~\ref{tab:nLarge_bLarge}), but   SCV revealed some advantage in Table~\ref{tab:pLarge_bLarge}. 
We made additional comparisons  between these two methods by varying the signal strength in the $p>n$ experiments; see      Table~\ref{tab:pLarge_highSNR}.
SCV produced the lowest standard errors and seemed to be more successful at picking models with the lowest MSE's, as well as  reducing instability.

\begin{table}
\caption{\label{tab:pLarge_highSNR} Prediction error comparison between PIC and  SCV in the $p>n$ experiment with different signal strengths.   Reported numbers are        MSE's with standard errors in parentheses (scaled for ease of comparison).}
\centering
\begin{tabular}{ l c c c c c c}
\hline
& \multicolumn{3}{c}{$\rho=.1$} & \multicolumn{3}{c}{$\rho=.5$}  \\
\cmidrule(lr){2-4} \cmidrule(lr){5-7}
    & $b=2$ & $b=3$ & $b=4$ & $b=2$ & $b=3$ & $b=4$ \\
\hline
\textbf{PIC} & 42 \footnotesize{(22)} & 42 \footnotesize{(23)} & 43 \footnotesize{(23)} & 28 \footnotesize{(19)} & 29 \footnotesize{(19)} & 29 \footnotesize{(19)}   \\
\textbf{SCV} & 34 \footnotesize{(16)} & 33 \footnotesize{(17)} & 33 \footnotesize{(17)} & 21 \footnotesize{(10)} & 21 \footnotesize{(10)} & 21 \footnotesize{(10)} \\
\hline
\end{tabular}
\end{table}%

Overall, in almost all cases either the PIC or SCV has the lowest prediction error, which is the ultimate goal of this work. In the weak signal strength situations, it could be argued the AIC is better at selection because it misses the fewest variables in terms of the noise-free simulation truth. However, as previously discussed, because the low SNR data are heavily contaminated by noise, variable selection as measured by the M and FA-rates may not be that meaningful; a more parsimonious \emph{and} predictive model may exist that diverges from the     zero-noise model used in generating synthetic data.

 \subsection{Yeast cell cycle data}

In an experiment conducted by~\cite{spellman1998comprehensive}, $106$ transcription factors (TFs) (also known as DNA binding proteins) were collected for $800$ yeast genes that regulate RNA levels within the eukaryotic cycle. The cell cycle was measured by taking RNA levels on the $800$ genes at $18$ time points using the $\alpha$ factor arrest method. In this data analysis we use a subset of the original dataset obtained from the R~\citep{Rlanguage} package ``spls"~\citep{chun2010sparse}. The $\bsbX$ matrix consists of $106$ transcription factors (TF) collected on $542$ genes; 
the $\bsbY$ matrix contains RNA levels measured on the same subset of genes at $18$ time points. 
For the dataset, there are $21$ experimentally verified TFs related to cell cycle regulation~\citep{wang2007group}.
They can serve as a biological truth and should be consistently picked by a variable selection technique. We centered and scaled both $\bsbX$ and $\bsbY$, performed selective reduced rank regression \citep{she2017selective} 
and compared the selection performance of ordinary  CV  and structural cross-validation (SCV). (The prediction performance of CV and SCV was found to be quite similar over 200 repeated training-test splits with 50\% for training and 50\% for test.)



We bootstrapped the data $200$ times to measure the stability of the two methods in terms of both rank and variable selection. The bootstrap distributions of the estimated ranks and cardinalities are displayed in Figure~\ref{fig:boxplots}. SCV's median $\hat{J}$ and $\hat{r}$ were equal to $86$ and $4$ and CV's median $\hat{J}$ and $\hat{r}$ were equal to $46$ and $7$, respectively. 
Clearly, the spread of the optimal ranks and cardinalities selected by CV is much larger than SCV, showing that CV is   unstable in terms of model selection. Even though CV's median $\hat{J}$ is smaller, its large variance suggests that this method picks and leaves out TFs in a more random fashion than does SCV. The number of free parameters is determined \emph{jointly} by $r$ and $J$; the median of the degrees of freedom was $400$ for SCV and $410$ for CV, suggesting that SCV does tend to select smaller models.  
In addition, CV showed a five-fold increase in computation time than SCV.


\begin{figure}[h]
\centering
\includegraphics[width=.7\textwidth, height=2.4in]{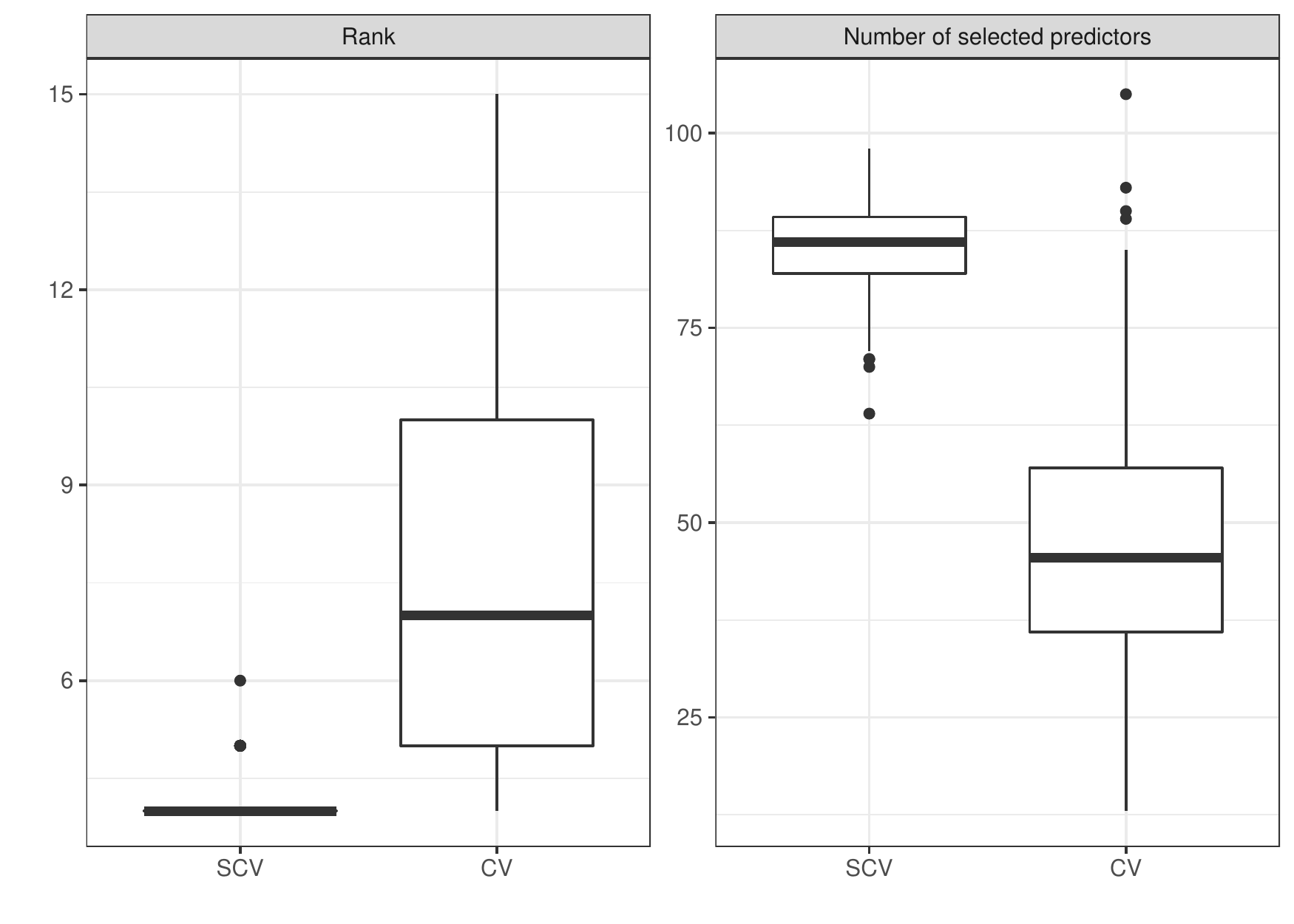} 
\caption{Rank and cardinality distributions on the bootstrapped data.}
\label{fig:boxplots}
\end{figure}

To assess the selection performance in terms of the $21$ experimentally verified TFs, the left panel of Figure~\ref{fig:TF21_freq} shows the percentage of bootstrap replicates in which each TF was selected. Each point corresponds to a TF and   the dotted line labels all identical selection frequencies by  the two methods. Notably, \textit{every} TF lies either on or above the line, showing that SCV's selection frequencies for the confirmed TFs were uniformly larger than CV's.
It is clear that CV often fails to select all of the TFs confirmed to be related to cell cycle regulation. For example, BAS1 is likely to be a significant regulator of the yeast cell cycle~\citep{cokus2006modelling} and was selected by SCV in nearly $75\%$ of the replicates, but CV selected this TF in under $25\%$ of the replicates. The right panel of Figure~\ref{fig:TF21_freq} tabulates the number of confirmed TFs selected at various percentage cut-offs of the bootstrap replicates. The large gap between SCV and CV shows that SCV is far more successful at selecting the verified TFs at \emph{any} cut-off. For example, roughly $50\%$ of the confirmed TFs are selected by CV in at least $50\%$ of the replicates, but SCV has a nearly $100\%$ success rate at the same cut-off.
SCV also identified some TFs that are not part of the confirmed subset. For example, SCV selected SKO1   $194$ times while CV picked  this TF less than 50\% of the time;~\cite{niu2008mechanisms} experimentally determined that overexpression of this TF is related to cell cycle progression.

\begin{figure}[h]
\centering

\includegraphics[width=.49\textwidth, height=2.4in]{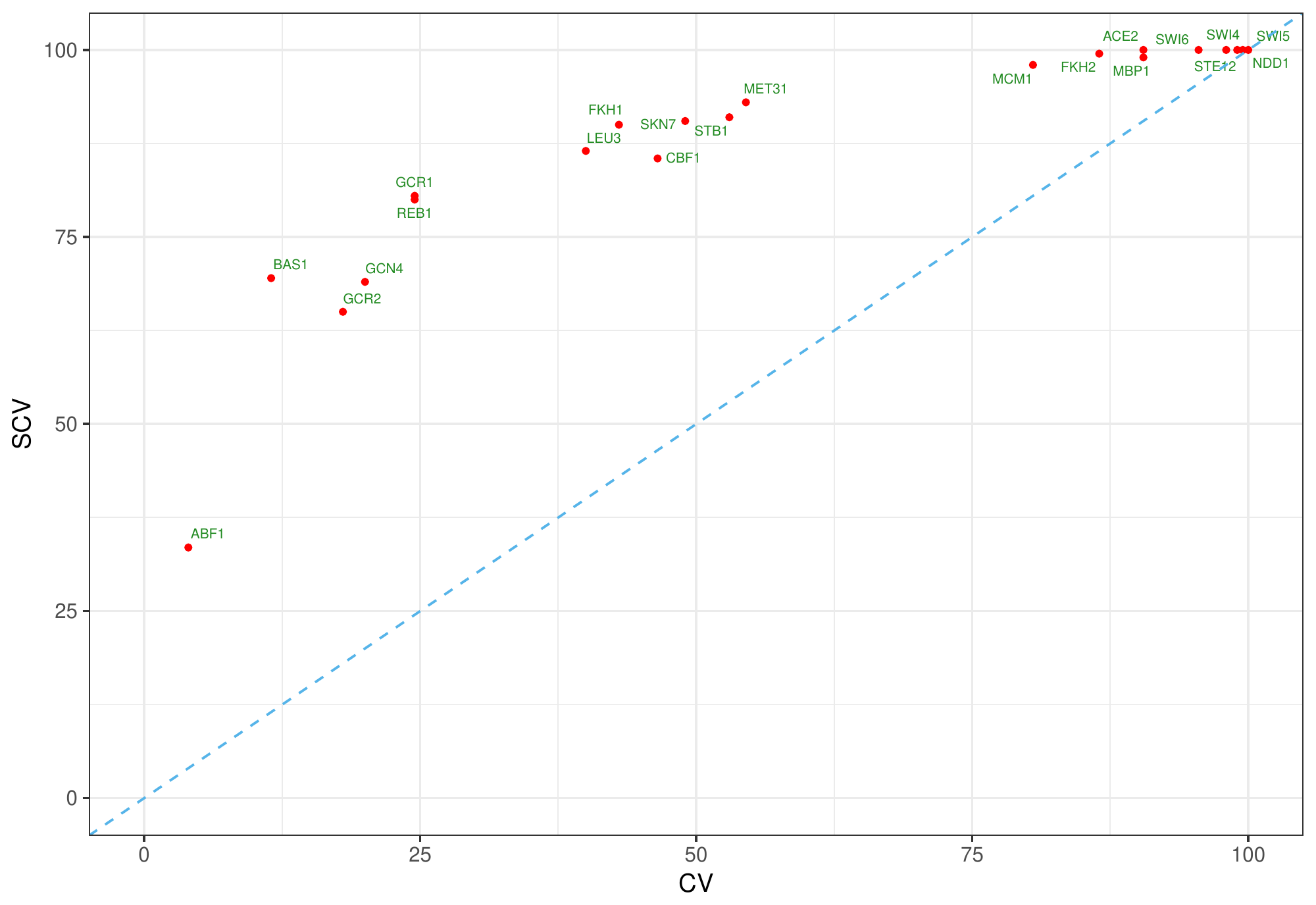}
\includegraphics[width=.49\textwidth, height=2.4in]{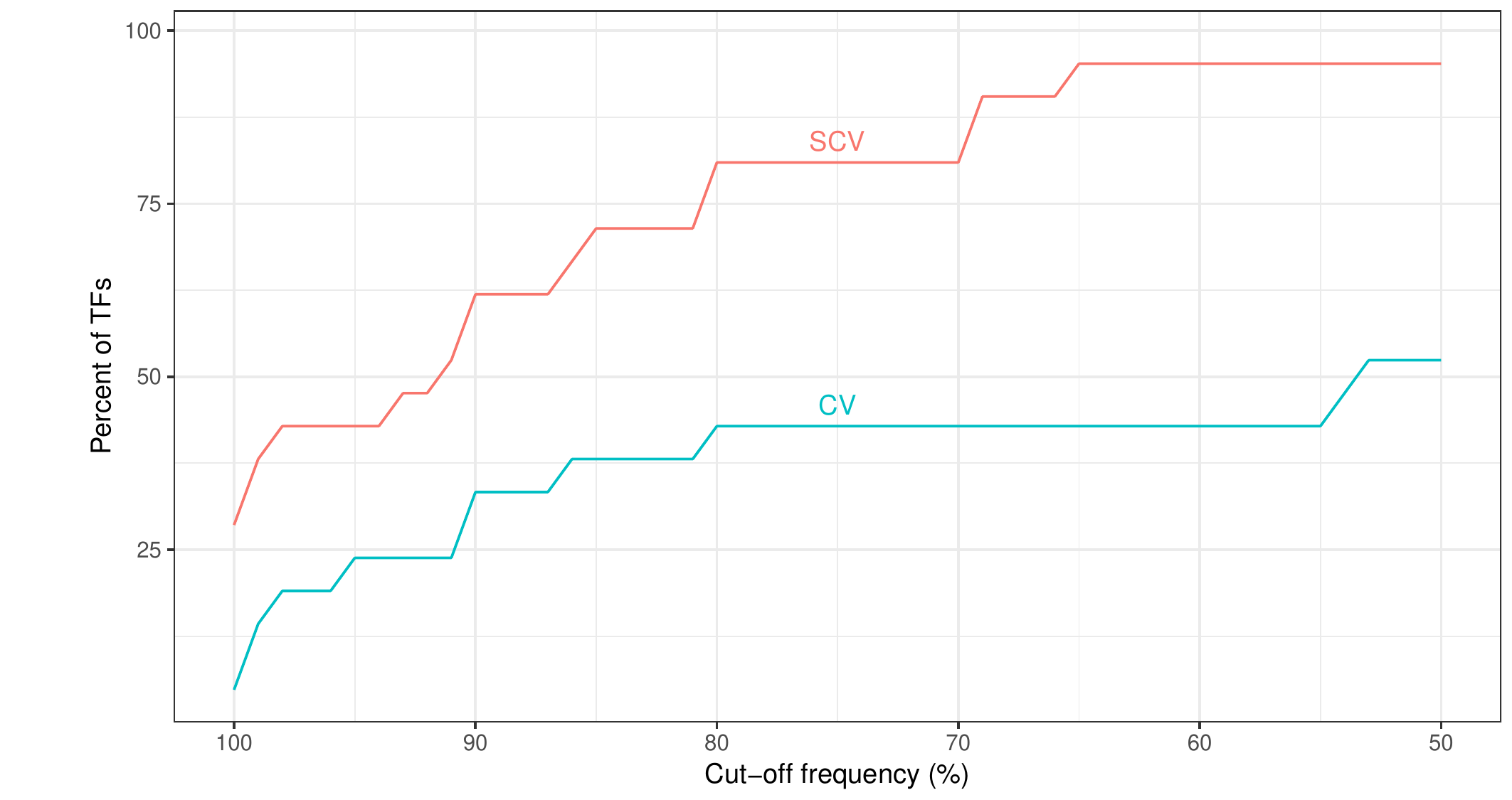}
\caption{Bootstrap selection frequencies of the $21$ experimentally verified TFs.}
\label{fig:TF21_freq}
\end{figure}

\section{Discussion}
The conventional approach to cross-validation involves splitting the data into $K$ subsets and calling a learning algorithm $K$ times. This procedure can be expensive and unstable because it may result in $K$ fitted models not  directly comparable. As a remedy to this problem, we proposed structural cross-validation, which maintains the same model in data resampling and is computationally efficient. Theoretically, we showed that the optimal complexity rate    for joint variable and rank selection is achieved  by the predictive information criterion non-asymptotically. Based on an identity built for cross-validation error, we proposed   scale-free rate correction for the commonly used $K$-fold cross-validation to match the optimal error rate. To the best of our knowledge, this calibrated cross-validation is novel.

The notions of structural cross-validation and PIC are applicable to either pure variable selection or pure rank selection. For example, in the case of a single response variable, the SCV in~\eqref{ratecorrSCV} takes the form $ {\mbox{CV-Err}}+ \alpha(\mbox{Trn-Err}/n)   J\log(e p/J)$ for some positive constant $\alpha$, whereas in   pure ridge-type problems  the inflation term disappears. The extraction of $s$-patterns extends  to   generalized lasso problems that  pursue   sparsity of $\bsbT \bsbb$ for a thin $\bsbT$ matrix \citep{SheSpaReg,TT11}. For example, given an estimate $\hat \bsbb$ satisfying $\bsbT[j,]\bsbb\ne 0$ for all $j\in \mathcal J$, or $\bsbT[{\mathcal J}^c, ]\hat\bsbb = \bsb0$,   we can construct a structural pattern $\bsbO \in \mathbb O^{p\times r}$, with $r=p-r(\bsbT[{\mathcal J}^c, ])$, that spans the orthogonal complement of the row space of $\bsbT[{\mathcal J}^c, ]$.

This framework pursues prediction accuracy 
as its ultimate principle and leads to some general theorems without large-$n$ assumption or incoherent design   conditions. 
If the  SNR  is not too small,  adopting the predictive learning    perspective also automatically implies faithful support recovery. 
Indeed, when the noise contamination was relatively small, our simulations showed that the proposed methods    had   low prediction error   \textit{and} satisfactory selection performance. But our experiments  also demonstrated  that when the SNR is small, it may not be valid to assess variable selection in reference to the noise-free simulation truth. These studies suggest that the aphorism ``all models are wrong but some are useful"~\citep{box1979robustness} seems to apply  in the small SNR scenarios.



\appendix

\section{Proof of Theorem \ref{th:minimax}}
We consider three cases as follows.

\textit{Case (i):} $( J  +m - r ) r \ge  J\log (e p / J  )$ and  $q\ge J $. Consider a signal  subclass \begin{align*}
{\mathcal B}^1(r)=\{\bsbB =& [b_{jk}],     \ b_{jk}=0 \mbox{ or } \gamma R \mbox{ if }  1\le j \le J, 1\le k \le r/2 \mbox{ or } 1\le j \le r/2 \\ & \mbox{ and }   b_{jk}=0 \mbox{ otherwise} \}.
\end{align*}
where $R= \frac{\sigma}{{\overline{\kappa}^{1/2}}  } $ and $\gamma>0$ is a small constant to be chosen later.
Clearly,  $|{\mathcal B}^1 (r)|= 2^{(J+m-r/2)r/2}$, $\mathcal B^1 (r)\subset S(J, r)$, and $r(\bsbB_1 - \bsbB_2)\leq r$, $\forall \bsbB_1, \bsbB_2 \in \mathcal B^1(r)$. Also, since $r\le (  J + m)/2$, $(J+m-r/2)r/2\geq c(J+m-r)r$ for some universal constant $c$.

Let $\rho(\bsbB_1, \bsbB_2)=\|\vect(\bsbB_1) - \vect(\bsbB_2)\|_0$ be the Hamming distance. By the  Varshamov-Gilbert bound (cf. Lemma 2.9 in \cite{tsybakov2009introduction}),
there exists  a subset ${\mathcal B}^{10}(r)\subset {\mathcal B}^{1}(r)$ and $\bsb0 \in  {\mathcal B}^{10}(r)$ such that
\begin{eqnarray*}
\log | {\mathcal B}^{10}(r)| \geq c_1 r(J+m-r),
\end{eqnarray*}
and
\begin{eqnarray*}
\rho(\bsbB_1, \bsbB_2) \geq c_2 r(J+m-r), \forall \bsbB_1, \bsbB_{2} \in \mathcal B^{10},  \bsbB_1\neq \bsbB_2
\end{eqnarray*}
for some universal constants $c_1, c_2>0$.
Then $\| \bsbB_1 - \bsbB_2\|_F^2 = \gamma^2 R^2 \rho(\bsbB_1, \bsbB_2) \geq c_2 \gamma^2 R^2 (J+m-r) r$. It follows from the   conditional number assumption that
\begin{align}
\| \bsbX \bsbB_1 - \bsbX \bsbB_2  \|_F^2 \geq c_2 \underline{\kappa}   \gamma^2 R^2 (J+m-r)r \label{separationLBound}
\end{align}
 for any $\bsbB_1, \bsbB_{2} \in \mathcal B^{10}$,  $\bsbB_1\neq \bsbB_2$, where $\underline{\kappa}/\overline{\kappa}$ is a positive constant.

For Gaussian models, the  Kullback-Leibler divergence of $\mathcal {N}( \vect(\bsbX \bsbB_2),\sigma^2 \bsb{I}\otimes\bsbI)$ (denoted by $P_{\bsbB_2}$) from $\mathcal {N}( \vect(\bsbX \bsbB_1) ,\sigma^2 \bsb{I}\otimes \bsbI)$ (denoted by $P_{\bsbB_1}$) is $\mathcal K(\mathcal P_{\bsbB_1}, \mathcal P_{\bsbB_2}) =   \|\bsbX \bsbB_1 -  \bsbX \bsbB_2 \|_F^2/({2\sigma^2})$. Let $P_{\bsb{0}}$ be $\mathcal {  N}(\bsb{0}, \sigma^2 \bsb{I}\otimes \bsbI)$. By the assumption again,
for any $\bsbB: r(\bsbB)\leq r$, we have
\begin{align*}
\mathcal K(P_{\bsb{0}}, P_{\bsbB}) \leq \frac{1}{2\sigma^2}\overline{\kappa}   \gamma^2 R^2 \rho(\bsb{0}, \bsbB) \leq \frac{\gamma^2}{\sigma^2}   R^2 (J+m-r)r
\end{align*}
where we used $\rho(\bsbB_1, \bsbB_2) \leq (r/2)(J+m-r/2) \le r(J+m-r)$. Therefore,
\begin{align}
\frac{1}{|\mathcal B^{10}|}\sum_{\bsbB\in \mathcal B^{10}} \mathcal K(P_{\bsb{0}}, P_{\bsbB})\leq \gamma^2 r(J+m-r). \label{KLUBound}
\end{align}

Combining  \eqref{separationLBound} and \eqref{KLUBound} and choosing a sufficiently small value of  $\gamma$, we can apply Theorem 2.7 of \cite{tsybakov2009introduction} to get the desired  lower bound.

\textit{Case (ii):} $(q\wedge  J  +m - r ) r\ge J \log (e p / J  )$ and  $q<J $. Let $\underline{\kappa}   = \sigma_{\min}^2(\bsbX)$ and $ \overline{\kappa} =\sigma_{\max}^2(\bsbX)$. Note that   $\sigma_{\min}(\bsbX)$ is the smallest \emph{nonzero} singular-value of $\bsbX$ (and so the  restricted condition number assumption can hold even when $p>n$). Assume that the SVD of     $\bsbX $ is   given by  $  \bsbU \bsbD \bsbV^T$ with $\bsbD$ of size $q\times q$.   Define   $ \bsbZ:=\bsbU\bsbD$  and $  \tilde{\mathcal S} (  r) =\{ \bsbA  \in \mathbb R^{q\times m}: r(\bsbA)\leq r \}$. Then
for any estimator $\hat\bsbB$, and $\hat \bsbA= \bsbV^T \hat \bsbB$, we have  \begin{align*}
&\sup_{ \bsbB^*  \in \mathcal S(J, r)} \EP[\|\bsbX \bsbB^* - \bsbX \hat\bsbB  \|_F^2\geq c {P_o(J, r)}]\\
\geq & \sup_{  \bsbA^*  \in  \tilde{\mathcal S}(  r)} \EP[\|\bsbZ\bsbA^* - \bsbZ \hat\bsbA  \|_F^2\geq c {P_o(J, r)}],
\end{align*} because $\tilde {\mathcal S}(  r) \subset \{ \bsbV^T \bsbB: \bsbB\in S(J,r )\}$ under $q<J$. The new design matrix $\bsbZ$ has $q$ columns, and satisfies     $\underline{\kappa}  \|\bsbA\|_F^2\leq \|\bsbZ \bsbA\|_F^2  \leq \overline{\kappa}    \| \bsbA\|_F^2$
 for any $\bsbA $. 
 Repeating the argument in (i) gives  the result. 

\textit{Case (iii):} $  J \log (e p / J  )\ge  (  J  +m - r ) r  $ and $J\le q$. Construct
\begin{align*}
{\mathcal B}^3(J)=\{\bsbB  =[\bsbb_{1}, \ldots, \bsbb_p]^T:     \bsbb_j = \bsb{0}  \mbox{ or }  \gamma R  \bsb{1}   ,  J(\bsbB)\le J \},
\end{align*}
where $R= \frac{\sigma}{{\overline{\kappa}^{1/2}}  } (\frac{\log (ep /J)}{m})^{1/2} $ and $\gamma>0$ is a small constant.
 Clearly,  ${\mathcal B}^{3}(J) \subset {S}(J,r) $.
By Stirling's approximation, $\log |\mathcal B^{3}(J)|\geq \log { p \choose J}  \geq   J \log (  p/J)  \geq c J \log ( e p /J)$ for some universal constant $c$.
Due to  Lemma A.3 in   \cite{Rigollet11},  there exists  a subset ${\mathcal B}^{30}(J)\subset {\mathcal B}^{3}(J)$ such that
\begin{eqnarray*}
\log | {\mathcal B}^{30}(J)| \geq c_1  J \log ( ep/J)
\mbox{ and }
\rho(\bsbB_1, \bsbB_2) \geq c_2 Jm, \forall \bsbB_1, \bsbB_{2} \in \mathcal B^{30},  \bsbB_1\neq \bsbB_2
\end{eqnarray*}
for some universal constants $c_1, c_2>0$.
The afterward treatment follows the same lines as in (i) and the details are omitted.

\section{Proof of Theorem \ref{th:pic}}



The theorem can be derived from Theorem 2 of \cite{she2017selective}. For the sake of completeness, we provide its  proof, which will be used for proving other theorems.
Recall that a  random variable  $\xi$ is sub-Gaussian if   $\EP(|\xi|\geq t) \leq C e^{-c t^2}$  for any $t>0$ and some constants  $C, c>0$,  and  its  scale     is defined as $\sigma( \xi) =
\inf \{\sigma>0: \EE\{\exp(\xi^2/\sigma^2)\} \leq 2\}$;
  $\bsbxi\in \mathbb R^p$ is   a sub-Gaussian random vector with its scale  bounded by $\sigma$, if   $\langle \bsbxi, \bsba \rangle$ is sub-Gaussian and  $\sigma( \langle \bsbxi, \bsba \rangle )\leq \sigma \|\bsba \|_2$  for any $ \bsba\in \mathbb R^{p}$.
In this proof, given any index set  $\mathcal J\subset [p]:=\{1,\ldots, p\}$, $\Proj_{\bsbX_{\mathcal J}}$  is abbreviated to  $\Proj_{\mathcal J}$    when there is no ambiguity. Given any matrix $\bsbA$, we use $cs(\bsbA)$ and $rs(\bsbA)$ to denote its column space and row space, respectively.
%
%
Because $P_o(\bsbB) = \sigma^2 [\{q\wedge  J(\bsbB) +m - r(\bsbB)\} r(\bsbB) + J(\bsbB)\log \{e p / J(\bsbB)\}]$ only depends on $J(\bsbB)$ and $r(\bsbB)$, we also denote it  by  $P_o(J(\bsbB), r(\bsbB))$.
The optimality of $\hat \bsbB$ implies that for any $\bsbB=\bsbB_l$
\begin{align}
\frac{1}{2}\| \bsbX \hat \bsbB - \bsbX \bsbB^*\|_F^2 \leq \frac{1}{2}\| \bsbX \bsbB - \bsbX \bsbB^*\|_F^2 + A \sigma^2 P_o(\bsbB) -A \sigma^2 P_o(\hat \bsbB)
+ \langle \bsbE, \bsbX \hat \bsbB - \bsbX \bsbB \rangle. \label{firstineq}
\end{align}
Let $\bsbDelta = \hat \bsbB - \bsbB$, $\hat{\mathcal J}= \mathcal J(\hat\bsbB)$, $\mathcal J= \mathcal J(\bsbB)$,   $J=J(\bsbB)$, $\hat J =  J(\hat\bsbB)$,  $r = r(\bsbB)$, $\hat r = r(\hat \bsbB)$.
Denote the orthogonal projection onto the  row space   of $\bsbX_{{\mathcal J}} \bsbB_{{\mathcal J}}$ by  $\Proj_{rs}$      and its orthogonal complement by  $\Proj_{rs}^{\perp}$.
Decompose $\bsbX \bsbDelta$ as follows
\begin{equation}\begin{aligned}
 \bsbX \bsbDelta = &\bsbX \bsbDelta \Proj_{rs} + \bsbX \bsbDelta \Proj_{rs}^{\perp} \\
= & \Proj_{\mathcal J} \bsbX \bsbDelta \Proj_{rs} + \Proj_{\mathcal J}^{\perp} \bsbX \bsbDelta \Proj_{rs} + \bsbX_{\hat{\mathcal J}} \hat\bsbB_{\hat{\mathcal J}} \Proj_{rs}^{\perp} \\
= & \Proj_{\mathcal J} \bsbX \bsbDelta \Proj_{rs} + \Proj_{\mathcal J}^{\perp} \bsbX_{\hat{\mathcal J}} \hat\bsbB_{\hat{\mathcal J}} \Proj_{rs} + \bsbX_{\hat{\mathcal J}} \hat\bsbB_{\hat{\mathcal J}} \Proj_{rs}^{\perp}.
\end{aligned}\label{decompofXD}
\end{equation}
Then
$\| \bsbX \bsbDelta \|_F^2  = \| \Proj_{\mathcal J} \bsbX \bsbDelta \Proj_{rs}\|_F^2 + \|\Proj_{\mathcal J}^{\perp} \bsbX_{\hat{\mathcal J}} \hat\bsbB_{\hat{\mathcal J}} \Proj_{rs}\|_F^2 + \|\bsbX_{\hat{\mathcal J}} \hat\bsbB_{\hat{\mathcal J}}  \Proj_{rs}^{\perp}\|_F^2$ and so
\begin{align}
\langle \bsbE, \bsbX \bsbDelta  \rangle &=  \langle \bsbE, \Proj_{\mathcal J} \bsbX \bsbDelta \Proj_{rs} \rangle + \langle \bsbE, \Proj_{\mathcal J}^{\perp} \bsbX_{\hat{\mathcal J}} \hat\bsbB_{\hat{\mathcal J}}\Proj_{rs}\rangle + \langle \bsbE, \bsbX_{\hat{\mathcal J}} \hat\bsbB_{\hat{\mathcal J}} \Proj_{rs}^{\perp}\rangle.  \label{decomp4}
\end{align}

\begin{lemma}{\cite[Lemma 4]{she2017selective}}\label{emprocBnd}
Suppose $\vect(\bsbE)$ is  sub-Gaussian with mean zero and scale bounded by $\sigma$.  
Given $\bsbX\in \mathbb R^{n\times p}$, $1\leq J\leq p$, $1\leq r \leq J\wedge m$, define $\Gamma_{J, r} = \{\bsbDelta\in \mathbb R^{n\times m}: \|\bsbDelta\|_F\leq 1, r(\bsbDelta) \leq r, cs(\bsbDelta) \subset cs(\bsbX_{\mathcal J}) \mbox{ for some } \mathcal J: | \mathcal J|\le J\}$. Let $$P_o'(J, r) = (q\wedge  J ) r+(m - r) r + \log {p\choose J}.$$
Then for any $t\geq 0$,
\begin{align}
\EP [\sup_{\bsbDelta \in \Gamma_{J,r}} \langle \bsbE, \bsbDelta \rangle \geq t \sigma +  \{L  P_o'(J,r)\}^{1/2} \sigma] \leq C\exp(- ct^2),
\end{align}
where $L, C, c>0$ are universal constants.
\end{lemma}

Applying the lemma to the first term  on the right-hand side of \eqref{decomp4}, we have for any $a, b, a'>0$,
\begin{align*}
&  \langle \bsbE, \Proj_{\mathcal J} \bsbX \bsbDelta \Proj_{rs} \rangle - \frac{1}{a} \|\Proj_{\mathcal J} \bsbX \bsbDelta \Proj_{rs}\|_F^2 -  b L \sigma^2P_o(J, r)   \\
\leq &
\|\Proj_{\mathcal J} \bsbX \bsbDelta \Proj_{rs}\|_F  \langle \bsbE, \Proj_{\mathcal J} \bsbX \bsbDelta \Proj_{rs} / \|\Proj_{\mathcal J} \bsbX \bsbDelta \Proj_{rs}\|_F\rangle \\&- 2\sigma({{  b}/{a}})^{1/2} \|\Proj_{\mathcal J} \bsbX \bsbDelta \Proj_{rs}\|_F \{{L P_o(J, r)\}^{1/2}}  \\
\leq &
\frac{1}{a'} \|\Proj_{\mathcal J} \bsbX \bsbDelta \Proj_{rs}\|_F^2 +  \frac{a'}{4}
\sup_{1\leq J \leq p, 1\leq r \leq m\wedge J}  \sup_{\bsbDelta\in \Gamma_{J, r}} [\langle \bsbE, \bsbDelta \rangle - 2(b/a)^{1/2} \sigma\{{L P_o(J, r)\}^{1/2}}]_+^2\\
\equiv & \frac{1}{a'} \|\Proj_{\mathcal J} \bsbX \bsbDelta \Proj_{rs}\|_F^2 +  \frac{a'}{4} \sup_{1\leq J \leq p, 1\leq r \leq m\wedge J} R_{J,r}^2  \equiv  \frac{1}{a'} \|\Proj_{\mathcal J} \bsbX \bsbDelta \Proj_{rs}\|_F^2 +  \frac{a'}{4} R^2.
\end{align*}
We  claim that  that  $\EE R^2 \le C \sigma^2$  when choosing $4b >a$. In fact,
\begin{align*}
& \EP(R\geq t \sigma ) \\
\leq & \sum_{J=1}^p \sum_{r=1}^{m \wedge J} \EP(R_{J,r}\geq t\sigma )\\
 \leq & \sum_{J=1}^p \sum_{r=1}^{m \wedge J}  \EP [\sup_{\bsbDelta \in \Gamma_{J,r}}  \langle \bsbE, \bsbDelta \rangle -\sigma\{L  P_o'(J,r)\}^{1/2}\\& \qquad\qquad\ \geq t \sigma +   2(b/a)^{1/2}\sigma \{{L P_o(J, r)\}^{1/2}}-\{{L  P_o'(J,r)}\}^{1/2} ]\\
 \leq & \sum_{J=1}^p  \sum_{r=1}^{m \wedge J} C\exp(-c t^2) \exp[- c \{(2(b/a)^{1/2}-1)^2 L  P_o(J, r)\}] \\
 \leq &  C  \exp(-c t^2) \exp(-c\log p)\sum_{J=1}^p\sum_{r=1}^{m \wedge J} \exp[- c  L  P_o(J, r)\}] \\
 \leq & C \exp(-c t^2) p^{-c}.
\end{align*}

Similarly,   noticing that $r(\bsbX_{\hat{\mathcal J}} \hat\bsbB_{\hat{\mathcal J}} \Proj_{rs}^{\perp})\le r(\hat\bsbB_{\hat{\mathcal J}}) \le \hat r$ and $CS(\bsbX_{\hat{\mathcal J}} \hat\bsbB_{\hat{\mathcal J}}\allowbreak \Proj_{rs}^{\perp})\allowbreak\subset CS(\bsbX_{\hat {\mathcal J}})$, we have $\langle \bsbE, \bsbX_{\hat{\mathcal J}} \hat\bsbB_{\hat{\mathcal J}} \Proj_{rs}^{\perp}\rangle-  \|\Proj_{\hat {\mathcal J}} \bsbX \bsbDelta \Proj_{rs}^{\perp}\|_F^2/a -  b L P_o(\hat J,\hat r) \le  \|\Proj_{\mathcal J} \bsbX \bsbDelta \Proj_{rs}\|_F^2/{a'} \allowbreak+  {a'} R^2/4$. To handle the second term in   \eqref{decomp4},
we introduce the following lemma.

\begin{lemma}{\cite[Lemma 5]{she2017selective}}\label{emprocBnd2}
Suppose $\vect(\bsbE)$ is  sub-Gaussian with mean zero and  $\psi_2$-norm bounded by $\sigma$.  
Given $\bsbX\in \mathbb R^{n\times p}$, $1\leq J\leq p$, $1\le J'\le p$, $1\leq r \leq J\wedge m$, define $\Gamma_{J',J, r} = \{\bsbDelta\in \mathbb R^{n\times m}: \|\bsbDelta\|_F\leq 1, r(\bsbDelta) \leq r,  cs(\bsbDelta) \subset cs(\Proj_{{\mathcal J}'}^{\perp} \Proj_{{\mathcal J}})   \mbox{ for some }  {\mathcal J}', {\mathcal J} \subset [p] \mbox{ satisfying }  | {\mathcal J}'|\allowbreak \le J',  | {\mathcal J}|\le J \}$. Let $$P_o''(J',J, r) =  \{q\wedge  J \wedge (p- J')\} r  +(m - r) r + \log {p\choose J} + \log {p\choose J'}.$$
Then for any $t\geq 0$,
\begin{align*}
\EP [\sup_{\bsbDelta \in \Gamma_{J',J,r}} \langle \bsbE, \bsbDelta \rangle \geq t \sigma +   \{L  P_o''(J',J,r)\}^{1/2} \sigma] \leq C\exp(- ct^2),
\end{align*}
where $L, C, c>0$ are universal constants.
\end{lemma}

 Because    $P_o''( J, \hat J, \hat r) \leq   P_o(J, r) + P_o(\hat J, \hat r)$, 
\begin{align*}
& \langle \bsbE, \Proj_{\mathcal J}^{\perp} \bsbX_{\hat{\mathcal J}} \hat\bsbB_{\hat{\mathcal J}}\Proj_{rs}\rangle - \frac{1}{a} \|\Proj_{\mathcal J}^{\perp}\bsbX_{\hat{\mathcal J}} \hat\bsbB_{\hat{\mathcal J}}\Proj_{rs}\|_F^2-  b L\sigma^2 (   P_o(J, r) + P_o(\hat J, \hat r))   \\\le & \
 \frac{1}{a'} \|\Proj_{\mathcal J}^{\perp} \bsbX \bsbDelta \Proj_{rs}\|_F^2 +  \frac{a'}{4} {R'}^2 , \end{align*}
where  $\EE {R'}^2 \leq C \sigma^2$ and   $L$ is a sufficiently large constant.
In summary, we have
\begin{align*}
&\langle \bsbE, \bsbX \bsbDelta \rangle
\\ \leq & (\frac{1}{a}+\frac{1}{a'}) \|\bsbX \bsbDelta\|_F^2  + b L \sigma^2\{P_o(J,  r) + P_o(J, r) + P_o(\hat J, \hat r) +  P_o(\hat J,    \hat r)\}  + \frac{a'}{4} ( 2 R^2 +  {R'}^2)\\
\leq & (\frac{1}{a}+\frac{1}{a'})(1+b') \|\bsbX \bsbB - \bsbX \bsbB^*\|_F^2 +  (\frac{1}{a}+\frac{1}{a'})(1+\frac{1}{b'}) \|\bsbX \hat \bsbB - \bsbX \bsbB^*\|_F^2 \\ & + 2 b L\sigma^2 \{P_o(J,   r) + P_o(\hat J,    \hat r)\}  +  \frac{a'}{4} ( 2 R^2 +  {R'}^2).
\end{align*}
Choose constants   $a$, $a'$, $b$, $b'$, and $A$ sufficiently large such that $({1}/{a}+{1}/{a'})(1+{1}/{b'})<{1}/{2}$,  $4  b>a$,
and  $A > 2 b  L$.   Setting $\bsbB = \bsbB_l$ (which is nonzero), and taking expectation on both sides, we obtain the  inequality as desired.

\section{Proof of Theorem \ref{th:sf-pic}}
We take \eqref{gcv-pic}  as an instance  to show the conclusion.
First,    $1/(1 -  \delta(\hat \bsbB))^2 \ge 1/(1 -   \delta(\hat \bsbB))$ due to   $0\le \delta(\hat\bsbB)<1$. Let $A_0$ and $A$ be chosen such that   $2A_0>A$ and so   $\delta(\bsbB^*)< 1/2$. Then       $1/(1 -2   \delta(  \bsbB^*))\ge 1/(1 -   \delta(  \bsbB^*))^2$.

Let $h(\bsbB; A) = 1/\{mn - A P_o (\bsbB)\}$.
Based on the previous facts and the  optimality of   $\hat \bsbB$, we have \begin{align*}
 \| \bsbY - \bsbX \hat \bsbB\|_F^2 \, h(\hat \bsbB; A) mn
  \le & \frac{\| \bsbY - \bsbX \hat \bsbB\|_F^2}{(1-\delta(\hat \bsbB))^2}  \\\le&  \frac{\| \bsbY - \bsbX   \bsbB^*\|_F^2}{(1-\delta(  \bsbB^*))^2} \\
\le & \| \bsbY - \bsbX   \bsbB^*\|_F^2 \, h(  \bsbB^*; 2A) mn,
\end{align*}
and so $\| \bsbY - \bsbX \hat \bsbB\|_F^2\le \| \bsbY - \bsbX   \bsbB^*\|_F^2 \, h(  \bsbB^*; 2A)/ h(\hat \bsbB; A)$ since $h(\hat\bsbB; A)> 0$. It follows that
\begin{align*}
& \| \bsbX \hat \bsbB - \bsbX \bsbB^*\|_F^2 \\
\leq & \|\bsbE\|_F^2 \{{h(\bsbB^*; 2A)}/{h(\hat\bsbB;  A)}-1\} + 2\langle \bsbE, \bsbX \hat \bsbB - \bsbX \bsbB^*\rangle \\
\leq & \frac{A  \|\bsbE\|_F^2}{mn\sigma^2 - 2A P_o(\bsbB^*)}\sigma^2 P_o(\bsbB^*) - \frac{A  \|\bsbE\|_F^2}{mn \sigma^2}\sigma^2P_o(\hat\bsbB) + 2\langle \bsbE, \bsbX \hat \bsbB - \bsbX \bsbB^*\rangle
\end{align*}

The stochastic term $ 2\langle \bsbE, \bsbX \hat \bsbB - \bsbX \bsbB^*\rangle$ can be decomposed and bounded in the same way as in the proof of Theorem \ref{th:pic}, except that     we use the  high-probability form  results here. For example, for   term $\langle \bsbE, \Proj_{\mathcal J} \bsbX \bsbDelta \Proj_{rs} \rangle$   in \eqref{decomp4},  Lemma \ref{emprocBnd} shows  that for any constants $a, b, a'>0$ satisfying $4b >a$, the following event
\begin{align*}
&  \langle \bsbE, \Proj_{\mathcal J} \bsbX \bsbDelta \Proj_{rs} \rangle \leq  ({1}/{a} +{1}/{a'}) \|\Proj_{\mathcal J} \bsbX \bsbDelta \Proj_{rs}\|_F^2 +  b L\sigma^2 P_o(J, r)
\end{align*}
 occurs with   probability at least
$
1-  \sum_{J=1}^p \sum_{r=1}^{m \wedge J} C  \exp[- c \{(2({{b}/{a}})^{1/2}-1\}^2 L\allowbreak P_o(J, r)/\sigma^2]$ or $1-C p^{-c}$ for a sufficiently large value of $L$. Repeating the analysis in the proof of Theorem \ref{th:pic} shows
$$
 2\langle \bsbE, \bsbX  \hat \bsbB - \bsbX \bsbB^* \rangle
\leq2  ({1}/{a}+{1}/{a'})  \|\bsbX \hat \bsbB - \bsbX \bsbB^*\|_F^2   + 4b L \sigma^2\{P_o(\hat\bsbB) + P_o(\bsbB^*)\},
$$
with probability at least $1-C p^{-c}$ for some $c, C>0$.

Let $\gamma$ and $\gamma'$ be  constants satisfying $0< \gamma <\ 1, \gamma'>0$. On $\mathcal E =\{(1-\gamma) {mn\sigma^2}\leq \|\bsbE\|_F^2 \leq (1+\gamma') {mn\sigma^2} \}$ ,  we have
\begin{align*}
& \frac{ \sigma^2P_o(\bsbB^*)A  \|\bsbE\|_F^2}{mn\sigma^2 - 2A \sigma^2P_o(\bsbB^*)}  - \sigma^2P_o(\hat\bsbB)\frac{A  \|\bsbE\|_F^2}{mn \sigma^2} \\ \leq \ &  \frac{(1+\gamma') A A_0}{A_0 - 2A}\sigma^2P_o(\bsbB^*)-   {(1-\gamma) A\sigma^2 }P_o(\hat\bsbB).
\end{align*}
From \cite{lm2000}, the complement of $\mathcal E$ occurs with probability at most $C'\exp(-c' m n)$ (with $c', C'$ dependent on $\gamma, \gamma'$).
 With $A_0$   large enough, we can choose  $a, a', b, A$ such that    $({1}/{a}+{1}/{a'})<{1}/{2}$,  $4  b>a$, and $4b L\leq ( 1-\gamma)A$. The conclusion results, and it is easy to see that it holds for the fractional form  as well.

Noticing that
$1/(1-\delta) \ge \exp(\delta) \ge 1+\delta \ge 1/(1-\delta/2) $ for any $0\le \delta < 1$, we can show the same conclusion   for the log-form and the plug-in form, the proof of which  follows the same lines.
\section{Proofs of Theorem \ref{th:iden} and Corollary \ref{iden-cor}}
 For   notational clarity, we prove the conclusion for $K=2$ only.  The proof in the general case follows the same lines. Let $\bar r = r$ and    $\bsbO=\bsbS \bsbU$ (which has    $r$ orthogonal columns).  Denote the overall data by $(\bar \bsbX, \bar \bsbY)$.  Let $(\bsbX, \bsbY)$ be the training data and   $(\tilde\bsbY, \tilde \bsbX)$   be the test data  in the first round.
Then      $\hat \bsbB=\bsbO\hat\bsbC$ and   $\hat \bsbC  = (\bsbX_{\bsbO}^T \bsbX_{\bsbO})^{-1} \bsbX_{\bsbO}^T \bsbY$ (where
$\bsbX_{\bsbO}$ is short for   $\bsbX \bsbO$ throughout the proof).   Also, from the model assumption, we can write     $\bsbB^* = \bsbO^* \bsbC^*$ and $\bsbX \bsbB^* = \bsbX_{\bsbO^*} \bsbC^*$ for some $\bsbC^*$, where  $\bsbO^*$ is the $s$-pattern of $\bsbB^*$.

First, we  calculate the  extra-sample test error on   $(\tilde\bsbY, \tilde \bsbX)$:
\begin{align*}
&\EE[\|\tilde\bsbY - \tilde \bsbX \hat\bsbB\|_F^2 ] \\
=& \EE[\|\tilde\bsbY - \tilde \bsbX_{\bsbO} \hat\bsbC\|_F^2 ]\\
= &\EE[\|\tilde\bsbX \bsbB^* + \tilde\bsbE - \tilde \bsbX_{\bsbO} (\bsbX_{\bsbO}^T \bsbX_{\bsbO})^{-1} \bsbX_{\bsbO}^T \bsbX_{\bsbO^*}\bsbC^* - \tilde \bsbX_{\bsbO} (\bsbX_{\bsbO}^T \bsbX_{\bsbO})^{-1} \bsbX_{\bsbO}^T \bsbE\|_F^2 ]\\
=& \EE[\|\tilde \bsbE\|_F^2 ] +\EE[ \|\tilde\bsbX_{\bsbO} (\bsbX_{\bsbO}^T \bsbX_{\bsbO})^{-1} \bsbX_{\bsbO}^T \bsbE\|_F^2]   +  \EE [\|\tilde \bsbX (\bsbB^* -   \bsbO (\bsbX_\bsbO^T \bsbX_\bsbO )^{-1} \bsbX_\bsbO^T \bsbX \bsbB^* )\|_F^2]]\\
=& md \sigma^2+\EE[ \|\tilde\bsbX_{\bsbO} (\bsbX_{\bsbO}^T \bsbX_{\bsbO})^{-1} \bsbX_{\bsbO}^T \bsbE\|_F^2]   +  \EE [\|\tilde \bsbX (\bsbB^* -   \bsbO (\bsbX_\bsbO^T \bsbX_\bsbO )^{-1} \bsbX_\bsbO^T \bsbX \bsbB^* )\|_F^2]],
\end{align*}
where we used $\EE[\bar\bsbE \bar\bsbE^T] = m\sigma^2 \bsbI$ and the independence between $\bar \bsbE$ and $\bar \bsbX$.
The second term can be simplified further:
\begin{align*}
& \EE[ \|\tilde\bsbX_{\bsbO} (\bsbX_{\bsbO}^T \bsbX_{\bsbO})^{-1} \bsbX_{\bsbO}^T \bsbE\|_F^2]\\   = & \EE  Tr[ (\bsbX_{\bsbO}^T \bsbX_{\bsbO})^{-1} \bsbX_{\bsbO}^T \EE[\bsbE \bsbE^T] \bsbX_{\bsbO} (\bsbX_{\bsbO}^T \bsbX_{\bsbO})^{-1} \tilde\bsbX_{\bsbO}^T\tilde\bsbX_{\bsbO}] \\
 = & m \sigma^2 \EE   Tr [(\bsbX_{\bsbO}^T \bsbX_{\bsbO})^{-1}  \tilde\bsbX_{\bsbO}^T\tilde\bsbX_{\bsbO}]  \\
 = &  m   \sigma^2 \EE  Tr \{(\bsbX_{\bsbO}^T \bsbX_{\bsbO})^{-1}  \EE[\tilde\bsbX_{\bsbO}^T\tilde\bsbX_{\bsbO}]\}  \\
  = &   m d  \sigma^2 \EE   Tr \{(\bsbX_{\bsbO}^T \bsbX_{\bsbO})^{-1}   {\bsbO}^T \bsbSig {\bsbO}\} \\
 = & m d  \sigma^2 \EE   Tr \{(\bsbZ_{n-d}^T \bsbZ_{n-d})^{-1}  \}.
\end{align*}
Therefore, the total cross-validation error is given by
\begin{align*}
\mbox{{CV-Err}}
=&    mn\sigma^2 + m n \sigma^2 Tr\{  \EE[(\bsbZ_{n-d}^T \bsbZ_{n-d})^{-1}]  \}   \\  & +  \EE [\|\tilde \bsbX (\bsbB^* -   \bsbO (\bsbX_\bsbO^T \bsbX_\bsbO )^{-1} \bsbX_\bsbO^T \bsbX \bsbB^* )\|_F^2]] \\ & +  \EE [\|  \bsbX (\bsbB^* -   \bsbO (\tilde \bsbX_\bsbO^T \tilde \bsbX_\bsbO )^{-1} \tilde \bsbX_\bsbO^T \tilde \bsbX \bsbB^* )\|_F^2]].
\end{align*}
Similarly,   the overall training error is
\begin{align*}
\EE[\|\bar \bsbY -  \bar \bsbX \hat \bsbB\|_F^2 ] = &\EE[\|(\bsbI - \Proj_{\bar\bsbX_{\bsbO}})\bar \bsbX \bsbB^* \|_F^2] + \EE[\|(\bsbI - \Proj_{\bar \bsbX_{\bsbO}})\bsbE \|_F^2]\\
= & m (n -r) \sigma^2 + \EE[\|(\bsbI - \Proj_{\bar \bsbX_{\bsbO}})\bar \bsbX \bsbB^* \|_F^2].
\end{align*}
It remains to analyze
$
  \EE [\|\tilde\bsbX  \bsbB^* -  \tilde \bsbX\bsbO (\bsbX_\bsbO^T \bsbX_\bsbO )^{-1} \bsbX_\bsbO^T \bsbX \bsbB^* \|_F^2]]  +  \EE [\|\bsbX  \bsbB^* - \allowbreak   \bsbX\bsbO (\tilde\bsbX_\bsbO^T \tilde\bsbX_\bsbO )^{-1}\allowbreak \tilde\bsbX_\bsbO^T \tilde\bsbX \bsbB^* \|_F^2]]
 -\EE [\|\bar\bsbX  \bsbB^* -  \bar \bsbX_\bsbO (\bar\bsbX_\bsbO^T \bar\bsbX_\bsbO )^{-1} \bar\bsbX_\bsbO^T\bar \bsbX \bsbB^* \|_F^2]].
$

Define $\Proj_1 = \Proj_{\bsbO} \cap \Proj_{\bsbO^{*}}$, $\Proj_2 = \Proj_{1}^\perp\cap \Proj_{\bsbO^*} (=
 \Proj_{ \bsbO^*} \setminus ( \Proj_{\bsbO} \cap \Proj_{\bsbO^* } ))=\mathfrak O \mathfrak O^T$. Then    $\bsbB^* = \bsbO^* \bsbC^* =  (\Proj_1 + \Proj_2)\bsbB^*$, and so
\begin{align*}
&\bsbB^* - \bsbO (\bsbX_\bsbO^T \bsbX_\bsbO )^{-1} \bsbX_\bsbO^T \bsbX  \bsbB^* \\ = & (\Proj_1 + \Proj_2) \bsbB^* -
\bsbO (\bsbX_\bsbO^T \bsbX_\bsbO )^{-1} \bsbX_\bsbO^T \bsbX(\Proj_1 + \Proj_2) \bsbB^*\\  = &   \Proj_2  \bsbB^* - \bsbO(\bsbX_\bsbO^T \bsbX_\bsbO )^{-1} \bsbX_\bsbO^T \bsbX  \Proj_2  \bsbB^*.
\end{align*}
With  this fact, and using the shorthand notation $\bsbX_{1} = \bsbX \bsbO$,  $ \bsbX_2 = \bsbX {\mathfrak O}$, $\bsbB_2^*=   {\mathfrak O}^T \bsbB^*$, we get
\begin{align*}
  \tilde\bsbX  \bsbB^* -  \tilde \bsbX\bsbO (\bsbX_\bsbO^T \bsbX_\bsbO )^{-1} \bsbX_\bsbO^T \bsbX \bsbB^*
& =     \tilde\bsbX_2   \bsbB_2^* -   \tilde\bsbX_{1} (\bsbX_{1}^T \bsbX_{1} )^{-1} \bsbX_{1}^T \bsbX_2   \bsbB_2^*\\
& =    \left[\tilde\bsbX_{1} \ \tilde \bsbX_2\right] \left[ \begin{array}{c} -(\bsbX_{1}^T \bsbX_{1} )^{-1} \bsbX_{1}^T \bsbX_2\\ \bsbI \end{array}\right]    \bsbB_2^*,
\end{align*}
and
\begin{align*}
& \bar\bsbX  \bsbB^* -  \bar \bsbX_\bsbO (\bar\bsbX_\bsbO^T \bar\bsbX_\bsbO )^{-1} \bar\bsbX_\bsbO^T\bar \bsbX \bsbB^*
=   \left[ \bar\bsbX_{1 } \ \bar\bsbX_2\right] \left[ \begin{array}{c} -(\bar\bsbX_{1}^T \bar\bsbX_{1} )^{-1}\bar \bsbX_{1}^T \bar\bsbX_2\\ \bsbI \end{array}\right]    \bsbB_2^*. \end{align*}

Introduce $\bsbOmega_{1} = \bsbX_{1}^T \bsbX_{1}$,  $\bsbOmega_{2} = \bsbX_{2}^T \bsbX_{2}$,  $\bsbOmega_{12}=\bsbX_{1 }^T \bsbX_{ 2}$. 
In the following, we investigate
\begin{align}
\begin{split}
&\left[\begin{array}{cc} -\bsbOmega_{1}^{-1} \bsbOmega_{1 2} \\ \bsbI \end{array}\right]^T \left[\begin{array}{cc} \tilde  \bsbOmega_{1} &  \tilde\bsbOmega_{1 2} \\   \tilde\bsbOmega_{2 1} &  \tilde\bsbOmega_{ 2}   \end{array}\right]
\left[\begin{array}{cc} -\bsbOmega_{1}^{-1} \bsbOmega_{1 2} \\ \bsbI \end{array}\right] \\
+ &  \left[\begin{array}{cc} -\tilde\bsbOmega_{1}^{-1} \tilde\bsbOmega_{1 2} \\ \bsbI \end{array}\right]^T \left[\begin{array}{cc}    \bsbOmega_{1} &   \bsbOmega_{1 2} \\    \bsbOmega_{2 1} &  \bsbOmega_{ 2}   \end{array}\right]
\left[\begin{array}{cc} - \tilde\bsbOmega_{1}^{-1} \tilde\bsbOmega_{1 2} \\ \bsbI \end{array}\right]\\
-&  \left[\begin{array}{cc} -\bar\bsbOmega_{1}^{-1} \bar\bsbOmega_{1 2} \\ \bsbI \end{array}\right]^T \left[\begin{array}{cc}    \bar\bsbOmega_{1} &   \bar\bsbOmega_{1 2} \\  \bar  \bsbOmega_{2 1} & \bar \bsbOmega_{ 2}   \end{array}\right]
\left[\begin{array}{cc} - \bar\bsbOmega_{1}^{-1} \bar\bsbOmega_{1 2} \\ \bsbI \end{array}\right].
\end{split} \label{keyterm}
\end{align}

 Define $\bsbOmega_{1}^{-1}\bsbOmega_{1 2}=(\bsbX_{1}^T\bsbX_{1})^{-1} \bsbX_{1}^T\bsbX_2 =   \bsbR + \bsbDelta$,   where  \ $\bsbR = \EE [(\bsbX_{1}^T\bsbX_{1})^{-1} \bsbX_{1}^T\bsbX_2 ]$ (non-random) and  $\EE\bsbDelta = \bsb0$.
By symmetry, $\tilde \bsbOmega_{1}^{-1}\tilde \bsbOmega_{1 2} =  \bsbR + \tilde \bsbDelta$   with $  \EE\tilde \bsbDelta  = \bsb0$.
Represent  $\bar \bsbOmega_{1}^{-1}\bar \bsbOmega_{1 2}$ by $\bsbR + \bar \bsbDelta$ (but $\EE\bar\bsbDelta $ may not be $\bsb0$ in general).

Writing $$\left[\begin{array}{cc} -\bsbR-\bsbDelta  \\ \bsbI \end{array}\right] = \left[\begin{array}{cc} -\bsbR  \\ \bsbI \end{array}\right] + \left[\begin{array}{cc} -\bsbDelta  \\ \bsb0 \end{array}\right]$$ and using the additivity  of   $$\bar \bsbOmega_{1 } =\bsbOmega_{1 } + \tilde \bsbOmega_{1 },  \bar \bsbOmega_{2 } =\bsbOmega_{2 } + \tilde \bsbOmega_{2 },  \bar \bsbOmega_{1 2} =\bsbOmega_{1 2} + \tilde \bsbOmega_{1 2}, $$ \eqref{keyterm} can be reduced to        the sum  of
\begin{align*}
 & \phantom{00}  \left[\begin{array}{cc} -\bsbR  \\ \bsbI \end{array}\right]^T \left[\begin{array}{cc} \tilde  \bsbOmega_{1} &  \tilde\bsbOmega_{1 2} \\   \tilde\bsbOmega_{2 1} &  \tilde\bsbOmega_{ 2}   \end{array}\right]
\left[\begin{array}{cc} -\bsbDelta \\ \bsb0 \end{array}\right] + \left[\begin{array}{cc} -\bsbDelta \\ \bsb0 \end{array}\right]^T \left[\begin{array}{cc} \tilde  \bsbOmega_{1} &  \tilde\bsbOmega_{1 2} \\   \tilde\bsbOmega_{2 1} &  \tilde\bsbOmega_{ 2}   \end{array}\right]\left[\begin{array}{cc} -\bsbR  \\ \bsbI \end{array}\right]
\\
 & +   \left[\begin{array}{cc} -\bsbR \\ \bsbI \end{array}\right]^T \left[\begin{array}{cc}    \bsbOmega_{1} &   \bsbOmega_{1 2} \\    \bsbOmega_{2 1} &  \bsbOmega_{ 2}   \end{array}\right]
\left[\begin{array}{cc} - \tilde\bsbDelta \\ \bsb0 \end{array}\right]+ \left[\begin{array}{cc} - \tilde\bsbDelta \\ \bsb0 \end{array}\right]^T \left[\begin{array}{cc}    \bsbOmega_{1} &   \bsbOmega_{1 2} \\    \bsbOmega_{2 1} &  \bsbOmega_{ 2}   \end{array}\right]\left[\begin{array}{cc} -\bsbR \\ \bsbI \end{array}\right]
\\
  & -   \left[\begin{array}{cc} -\bsbR \\ \bsbI \end{array}\right]^T \left[\begin{array}{cc}    \bar\bsbOmega_{1} &   \bar\bsbOmega_{1 2} \\  \bar  \bsbOmega_{2 1} & \bar \bsbOmega_{ 2}   \end{array}\right]
\left[\begin{array}{cc} - \bar\bsbDelta \\ \bsb0 \end{array}\right]-\left[\begin{array}{cc} - \bar\bsbDelta \\ \bsb0 \end{array}\right]  ^T \left[\begin{array}{cc}    \bar\bsbOmega_{1} &   \bar\bsbOmega_{1 2} \\  \bar  \bsbOmega_{2 1} & \bar \bsbOmega_{ 2}   \end{array}\right] \left[\begin{array}{cc} -\bsbR \\ \bsbI \end{array}\right]
\\
= & \, 2 ( \bsbR^T \tilde \bsbOmega_{1}  - \tilde \bsbOmega_{2 1}) \tilde   \bsbDelta +  2 ( \bsbR^T  \bsbOmega_{1}  -  \bsbOmega_{2 1})   \tilde\bsbDelta - 2 ( \bsbR^T  \bar\bsbOmega_{1}  -  \bar\bsbOmega_{2 1})   \bar\bsbDelta
\end{align*}
and
\begin{align*}
&\phantom{+\ }\left[\begin{array}{cc} -\bsbDelta  \\ \bsb0 \end{array}\right]^T \left[\begin{array}{cc} \tilde  \bsbOmega_{1} &  \tilde\bsbOmega_{1 2} \\   \tilde\bsbOmega_{2 1} &  \tilde\bsbOmega_{ 2}   \end{array}\right]
\left[\begin{array}{cc} -\bsbDelta \\ \bsb0 \end{array}\right] \\
  &+  \left[\begin{array}{cc} -\tilde\bsbDelta \\ \bsb0 \end{array}\right]^T \left[\begin{array}{cc}    \bsbOmega_{1} &   \bsbOmega_{1 2} \\    \bsbOmega_{2 1} &  \bsbOmega_{ 2}   \end{array}\right]
\left[\begin{array}{cc} - \tilde\bsbDelta \\ \bsb0 \end{array}\right]\\
& -  \left[\begin{array}{cc} -\bar\bsbDelta \\ \bsb0 \end{array}\right]^T \left[\begin{array}{cc}    \bar\bsbOmega_{1} &   \bar\bsbOmega_{1 2} \\  \bar  \bsbOmega_{2 1} & \bar \bsbOmega_{ 2}   \end{array}\right]
\left[\begin{array}{cc} - \bar\bsbDelta \\ \bsb0 \end{array}\right]
\\= & \ \bsbDelta^T \tilde \bsbOmega_{1}\bsbDelta + \tilde \bsbDelta^T   \bsbOmega_{1}\tilde \bsbDelta -  \bar \bsbDelta^T   \bar \bsbOmega_{1}\bar \bsbDelta.
\end{align*}

In expectation,    $2 ( \bsbR^T \tilde \bsbOmega_{1}  - \tilde \bsbOmega_{2 1})   \bsbDelta$ and $ 2 ( \bsbR^T  \bsbOmega_{1}  -  \bsbOmega_{2 1})   \tilde\bsbDelta$ vanish. From the data splitting manner and the row independence, we have
\begin{align*}
&\EE Tr\{  \bsbB_2^{ * T}(\bsbDelta^T \tilde \bsbOmega_{1}\bsbDelta + \tilde \bsbDelta^T   \bsbOmega_{1}\tilde \bsbDelta ) \bsbB_2^*\}\\
= & Tr\{  \EE[\tilde \bsbOmega_{1}]\EE[\bsbDelta \bsbB_2^* \bsbB_2^{* T} \bsbDelta^T] +  \EE[ \bsbOmega_{1}]\EE[ \tilde \bsbDelta \bsbB_2^*\bsbB_2^{*T}\tilde \bsbDelta^T]\}\\
=& Tr\{  \EE[\bar \bsbOmega_{1}]\EE[\bsbDelta \bsbB_2^* \bsbB_2^{* T} \bsbDelta^T]\}  \\
 = &\EE Tr\{\bsbB_2^{* T}[((\bsbX_{1}^T\bsbX_{1})^{-1} \bsbX_{1}^T\bsbX_2 -   \bsbR)^T\bsbSig_{1} ((\bsbX_{1}^T\bsbX_{1})^{-1} \bsbX_{1}^T\bsbX_2 -   \bsbR) ]\bsbB_2^{*}\}.
\end{align*}
where $\bsbSig_{1} = \EE\bar \bsbOmega_{1}= \bsbO^T \EE[\bar \bsbX^T \bar \bsbX]\bsbO = \bsbO^T (n\bsbSig) \bsbO $.
  Finally,
\begin{align*}
&
-  \bar \bsbDelta^T   \bar \bsbOmega_{1}\bar \bsbDelta- 2 ( \bsbR^T  \bar\bsbOmega_{1}  -  \bar\bsbOmega_{2 1})   \bar\bsbDelta\\
 = & - (\bar\bsbOmega_{1}^{-1} \bar \bsbOmega_{1 2} - \bsbR)^T \bar\bsbOmega_{1}(\bar\bsbOmega_{1}^{-1} \bar \bsbOmega_{1 2} - \bsbR)+2 ( \bsbR^T \bar \bsbOmega_{1}   -  \bar \bsbOmega_{2 1})( \bsbR-\bar\bsbOmega_{1}^{-1} \bar \bsbOmega_{1 2}  ) \\
= &    \bsbR^T    \bar \bsbOmega_{1 } \bsbR -\bsbR^T \bar \bsbOmega_{1 2}   - \bar \bsbOmega_{1 2}^T \bsbR + \bar \bsbOmega_{1 2}^T \bar\bsbOmega_{1}^{-1} \bar \bsbOmega_{1 2} \\
=& (\bar \bsbOmega_{1}^{-1} \bar \bsbOmega_{1 2} - \bsbR)^T \bar \bsbOmega_{1} (\bar \bsbOmega_{1}^{-1} \bar \bsbOmega_{1 2}-\bsbR).
\end{align*}
The conclusion   follows.

To show the result in the corollary, use the fact that  under the Gaussian design assumption, $\bsbZ_{n-d}^T \bsbZ_{n-d}$ follows a Wishart distribution $\mathcal W_r(\bsbI, n - d)$, and  so $(\bsbZ_{n-d}^T \bsbZ_{n-d})^{-1}$ follows an inverse Wishart distribution and has mean  $\bsbI_{r}/(n - d -1)$.

\section{Rowwise error bound for support recovery}
\label{appsec:supp}
Given $\bsbB = [\bsbb_1, \ldots, \bsbb_p]^T\in \mathbb R^{p\times m}$, define $\|\bsbB\|_{2, \infty} = \max_{1\le j \le p } \|\bsbb_j\|_2$. Recall $J^* = J(\bsbB^*) = \| \bsbB^*\|_{2,0}$ and $r^* = r(\bsbB^*)$. 

\setcounter{theorem}{4}
\begin{theorem}\label{th:supp}
Assume  $\bsbY = \bsbX \bsbB^* + \bsbE$ where  $\vect(\bsbE)$ is sub-Gaussian with mean $\bsb0$ and scale bounded by $\sigma$ and $\bsbB^*\ne \bsb0$. Let $
\hat \bsbB = {\mathop{\rm arg\, min}}_{\bsbB  }\frac{1}{2}\| \bsbY - \bsbX \bsbB \|^2_F + P(\bsbB; A)$ where $P(\bsbB;A)=A \sigma^2 P_o(\bsbB)$ and $A$ is a constant. Assume that there exist $\kappa>0$ and large enough $A$ such that  $\kappa n  J^*  \|\bsbB  - \bsbB^*\|_{2, \infty}^2/2 \le   \|\bsbX \bsbB  - \bsbX \bsbB^*\|_F^2/2 + P(\bsbB ;A) + P(\bsbB^*; A)$ for any $\bsbB\in \mathbb R^{p\times m} $. Let   $\bsbB^* = [\bsbb_1^*, \ldots, \bsbb_p^*]^T$, $\hat \bsbB = [\hat \bsbb_1, \ldots, \hat \bsbb_p]^T$ and $\hat {\mathcal J} = \mathcal J(\hat \bsbB)$. Then
\begin{align}
\| \hat \bsbB - \bsbB^*\|_{2, \infty} \le 8A\sigma^2 \frac{P_o(\bsbB^*)}{\kappa n J^*} \le 8A\sigma^2 \left\{\frac{r^{*} + \log p}{n\kappa }+\frac{mr^*}{ n J^* \kappa} \right\}.\label{coorderror}
\end{align}
If, in addition, the minimum signal strength satisfies
\begin{align}
\min_{j\in \mathcal J^*} \|\bsbb_j^*\|_2 \ge 4 \sqrt {2A} \sigma \left \{ \frac{r^{*} + \log p}{n\kappa }+\frac{mr^*}{ n J^* \kappa}  \right\}^{1/2},\label{snrcond}
\end{align}
then with probability at least $1-Cp^{-c}$ for some constants $C, c>0$, $ {\mathcal J}^* \subset \hat {\mathcal J}$ and $\|\hat \bsbb_j\|_2$ for all    $j\in \mathcal J^*$  are exactly  the $J^*$ largest norms in $\|\hat \bsbb_j\|_2$ ($1\le j \le p$).
\end{theorem}

\begin{proof}
First, the stochastic term $ 2\langle \bsbE, \bsbX \hat \bsbB - \bsbX \bsbB^*\rangle$ can be decomposed and bounded in the same way as in the proof of Theorem \ref{th:pic}, except that     we use the  high-probability form  results here. For example, for   term $\langle \bsbE, \Proj_{\mathcal J} \bsbX \bsbDelta \Proj_{rs} \rangle$   in \eqref{decomp4},  Lemma \ref{emprocBnd} shows  that for any constants $a, b, a'>0$ satisfying $4b >a$, the following event
\begin{align*}
&  \langle \bsbE, \Proj_{\mathcal J} \bsbX \bsbDelta \Proj_{rs} \rangle \leq  ({1}/{a} +{1}/{a'}) \|\Proj_{\mathcal J} \bsbX \bsbDelta \Proj_{rs}\|_F^2 +  b L \sigma^2 P_o(J, r)
\end{align*}
 occurs with   probability at least
$
1-  \sum_{J=1}^p \sum_{r=1}^{m \wedge J} C  \exp[- c \{(2({{b}/{a}})^{1/2}-1\}^2 L\allowbreak P_o(J, r)/\sigma^2]$ or $1-C p^{-c}$ for a sufficiently large value of $L$. Repeating the analysis in the proof of Theorem \ref{th:pic} shows
$$
  \langle \bsbE, \bsbX  \hat \bsbB - \bsbX \bsbB^* \rangle
\leq   ({1}/{a}+{1}/{a'})  \|\bsbX \hat \bsbB - \bsbX \bsbB^*\|_F^2   + 2b L \sigma^2\{P_o(\hat\bsbB) + P_o(\bsbB^*)\},
$$
with probability at least $1-C p^{-c}$ for some $c, C>0$, from which it follows that \begin{align*}
(\frac{1}{2}- \frac{1}{a}-\frac{1}{a'})\| \bsbX \hat \bsbB - \bsbX \bsbB^*\|_F^2 \leq   (A+2bL)\sigma^2 P_o(\bsbB^{*}) -(A-2bL)  \sigma^2P_o(\hat \bsbB).
\end{align*}

From the assumption,
$$
 \kappa n  J^* \|\hat \bsbB - \bsbB^*\|_{2, \infty}^2 \le    \|\bsbX \hat \bsbB  - \bsbX\bsbB^*\|_F^2 + 2 P(\hat \bsbB;A)+ 2 P(\bsbB^*;A).
$$
Hence we obtain
$$
 \kappa n  J^* \|\hat \bsbB - \bsbB^*\|_{2, \infty}^2 + \frac{\frac{2}{a}+\frac{2}{a'}-\frac{2bL}{A}}{\frac{1}{2}- \frac{1}{a}-\frac{1}{a'}}  A\sigma^2P_o(\hat \bsbB)\le( 2+   \frac{1+\frac{2bL}{A}}{\frac{1}{2}- \frac{1}{a}-\frac{1}{a'}} )A\sigma^2P_{o}(\bsbB^*).
$$
Let $1/a + 1/a'=1/(2(1+\alpha))$, $4b> a$ and $A =   {2bL}(1+\alpha)$ with $\alpha>0$. Then
$$
 \|\hat \bsbB - \bsbB^*\|_{2, \infty}^2  \le\frac{1}{ \kappa n  J^*}( 2+  2(1+\frac{2}{\alpha}) ] (1+\alpha) 2bL\sigma^2P_{o}(\bsbB^*).
$$
Taking the optimal    $\alpha=1$, we get \eqref{coorderror}.

Furthermore, under \eqref{snrcond}, for any $j\in {\mathcal J}^*$
$$
\|\hat\bsbb_j \|_2 = \|\hat\bsbb_j  -   \bsbb_j^* + \bsbb_j^*\|_2\ge \|\bsbb_j^*\|_2  - \|\hat\bsbb_j  - \hat \bsbb_j^*\|_2 \ge2 \sqrt {2A} \sigma \left \{ \frac{r^{*} + \log p}{n\kappa }+\frac{mr^*}{ n J^* \kappa}  \right\}^{1/2}.  $$
Combining it with  $\|\hat \bsbb_{j'}\|_2^2 \le 8P(\bsbB^*;A)/(\kappa n J^*)$  for any $j'\in {{\mathcal J}^*}^c$ gives  the conclusion. 
\end{proof}

\bibliographystyle{apalike}\bibliography{crbib}

\end{document}